\documentclass{article}

\usepackage{custom-preamble}
\usepackage{custom-todo}
\usepackage{authblk}
\usepackage{orcidlink}
\usepackage[pagewise]{lineno}

\newcommand{\leo}[1]{\textcolor{black}{#1}}
\newcommand{\niels}[1]{\textcolor{black}{#1}}

\newcommand{\CorrB}[1]{\textcolor{black}{#1}}

\title{\CorrB{Bounds on the sequence length sufficient to reconstruct binary level-$1$ phylogenetic networks under the CFN model}}

\author[1,*]{Martin Frohn\orcidlink{0000-0002-5002-4049}}
\author[2]{Niels Holtgrefe\orcidlink{0009-0001-6162-9668}}
\author[2]{Leo van Iersel\orcidlink{0000-0001-7142-4706}}
\author[3]{Mark Jones\orcidlink{0000-0002-4091-7089}}
\author[1]{Steven Kelk\orcidlink{0000-0002-9518-4724}}

\affil[1]{Department of Advanced Computing Sciences, Maastricht University, Maastricht, 6229~EN, The Netherlands}
\affil[2]{Delft Institute of Applied Mathematics, Delft University of Technology, Delft, 2628~CD, The Netherlands}
\affil[3]{Department of Computer Science, Middlesex University, Hendon Town Hall, The Burroughs, London, NW4 4BT, United Kingdom}
\affil[*]{Corresponding Author. Email: martin.frohn@maastrichtuniversity.nl}

\begin{document}

\maketitle

%\linenumbers

\noindent\makebox[\linewidth]{\rule{\textwidth}{1pt}} \\

\begin{abstract}
Phylogenetic trees and networks are graphs used to model evolutionary relationships, with trees representing strictly branching histories and networks allowing for 
events in which lineages merge, called \emph{reticulation} events.
While the question of data sufficiency has been studied extensively in the context of trees, it remains largely unexplored for networks. In this work we take a first step in this direction
by establishing bounds on the amount of genomic data required to reconstruct binary level-$1$ semi-directed phylogenetic networks, which are binary networks in which reticulation events are indicated by directed edges, all other edges are undirected, and cycles are vertex-disjoint. For this class, methods have been developed recently that are statistically consistent. Roughly speaking, such methods are guaranteed to reconstruct the correct network assuming infinitely long genomic sequences. Here we consider the question whether networks from this class can be uniquely and correctly reconstructed from finite sequences.
Specifically, we present an inference algorithm that takes as input genetic sequence data, and demonstrate that the sequence length sufficient to reconstruct the correct network with high probability,
under the \CorrB{CFN model} of evolution, scales logarithmically, polynomially, or polylogarithmically with the number of taxa, depending on the parameter regime. As part of our contribution, we also present novel inference rules for quartet data in the semi-directed phylogenetic network setting.
\end{abstract}


\noindent\makebox[\linewidth]{\rule{\textwidth}{1pt}}

\section{Introduction}
Phylogenetic networks are graphs that model the evolutionary history of a set of taxa, for example different species, which are represented by the leaves of the network. Evolutionary events in which lineages merge, such as hybridization or recombination, are represented by so-called \emph{reticulation vertices} that have more than one parent. Rooted phylogenetic networks are directed acyclic graphs with a single root, see Figure~\ref{fig:preliminaries}(a). However, since the root location is not identifiable under most models~\niels{\cite{gross2018distinguishing,banos2019identifying}}, we consider semi-directed phylogenetic networks here, in which the incoming edges of reticulation vertices are directed, all other edges are undirected and the root is suppressed, see Figure~\ref{fig:preliminaries}(b). For formal definitions, see Section~\ref{sec:preliminaries}.

\leo{Recently, several methods have been developed for reconstructing semi-directed level-1 phylogenetic networks\footnote{\CorrB{In the directed setting, level-1 networks were originally called `galled trees'~\cite{gusfield2003efficient}.}}, which are networks in which cycles of the underlying undirected graph are vertex disjoint, from biological data consisting of gene trees or aligned DNA sequences~\cite{allman2019nanuq,allman2024nanuq+,holtgrefe2024squirrel,solis2016inferring,kong2025inference}. Such networks are relevant in practice because they allow for both speciation (divergence) and reticulation (convergence) events to explain the evolutionary relationships between taxa while assuming that reticulation events are (in some sense) isolated~\cite{bapteste2013networks,kong2022classes}. However, although some of the methods mentioned above are statistically consistent, i.e. guaranteed to reconstruct the correct network  given perfect data, these papers do not discuss how much data is needed. In particular, statistically-consistent methods for reconstructing networks from DNA sequences~\cite{holtgrefe2024squirrel,martin2025algebraic} may only be successful when provided sequences that are infinitely long.}

For phylogenetic trees, data sufficiency (also called \emph{sample complexity}) has been studied for various substitution models, starting with the seminal work in~\cite{erdos1999logs1,erdos1999logs2}, its improvements~\cite{king2003complexity,mossel2004phase,mossel2003impossibility}, variations~\cite{cs1999recovering,cryan2001evolutionary} and generalizations~\cite{warnow2001absolute,zhang2018new}. Recently, lower bounds on the sample complexity have also been established under a multi-species coalescent model~\cite{hill2025lower}. Typically, such reconstruction methods are \emph{fast converging}~\cite{huson1999disk}, i.e., %for
given lower and upper bounds on edge lengths and a DNA sequence alignment of polynomial size (in the number of taxa),
%the estimated edge lengths of a tree,
% LEO: what does ``estimated'' refer to?
the method recovers the true tree topology with high probability. A potential drawback of fast converging methods is that they may rely on known edge length bounds and absence of such bounds might even lead to statistical inconsistencies~\cite{huson1999disk,cs1999recovering,cryan2001evolutionary}. Methods that are fast converging without requiring any knowledge of the bounds on edge lengths are called \emph{absolute fast converging}~\cite{warnow2001absolute}. %A crucial characteristic of all absolute fast converging methods is their reliance on the four-point method~\cite{pereira1969note}. In our study of phylogenetic level-1 network we will stay within the same framework, and apply the four-point method to distinguish between tree- and network-like evolution.
%\todo{SK: do these convergence notions come back again later in the intro? DO we explain what the link is with our alg? NH: yes and yes :)} %A crucial characteristic of all absolute fast converging methods is their reliance on the four-point method~\cite{pereira1969note}. In our study of phylogenetic level-1 network we will stay within the same framework, i.e., we will differentiate between trees and cycles based on the application of the four-point method to estimated pairwise distances. 
%We also note that the sample complexity of tree reconstruction methods has been studied under the multispecies coalescence model~\cite{hill2025lower}.

%To date, the only result on data sufficiency for phylogenetic network inference is by~\cite{warnow2025advances}, who present a theoretical result that relies on access to an oracle indicating which DNA sites are \emph{homoplasy-free}, i.e., having changed state in exactly one edge of the network.

%who present a theoretical result that relies on access to an oracle indicating which DNA sites are \emph{homoplasy-free}, i.e., having changed state in exactly one edge of the network.

In a recent paper, a first result on data sufficiency for phylogenetic network inference was presented~\cite{warnow2025advances} focusing on  reconstructing binary semi-directed
%\footnote{Although \cite[Theorem~9]{warnow2025advances} claims that rooted networks are reconstructible from SNPs without knowing ancestral states, we believe that the paper only proves this for semi-directed networks.} 
level-1 phylogenetic networks from SNPs, basically a binary sequence alignment, in which each site (column) is assumed to have evolved down a tree displayed by the network in such a way that its state changed exactly once. The latter assumption, also called \emph{homoplasy-free} or \emph{infinite sites}, is a severely restricting condition on the data that we do not assume in this work.

In this article, we introduce a new %reliable 
method to reconstruct binary semi-directed level-$1$ phylogenetic networks from %limited available 
biological data under the \CorrB{
Cavender-Farris-Neyman (CFN) model}~\cite{neyman1971molecular}: the %symmetric
two-state analogue of the well-known Jukes-Cantor model.
We assume a substitution process along each network edge, without \CorrB{incomplete} lineage sorting, where at reticulation nodes each site is independently inherited from one parent according to inheritance probabilities. This model also accommodates homoplastic sites, allowing for convergent or parallel mutations. %that may obscure true evolutionary relationships.
Our method reconstructs binary semi-directed level-$1$ networks up to the suppression of $2$-cycles, $3$-cycles and reticulation vertices in $4$-cycles, as do other methods~\cite{allman2019nanuq,allman2024nanuq+, warnow2025advances}, since this is best possible \niels{for displayed-quartet-based methods}~\cite{banos2019identifying}.
%\niels{when relying only on displayed quartets}~\cite{banos2019identifying}. %In some cases under some models you can get triangles and 2-blobs (eg trinet inequality, or multiple samples). So, this is only true if we work with quartet topologies and puzzling those together.

\begin{table*}[!t]
\centering
    \begin{tabular}{l|ll}
     \toprule
        & \multicolumn{2}{c}{range of mutation probabilities $p(e)$, $e\in E(N)$} \\
     	$\max\left\{\text{depth}(N),lc(N)\right\}$ &  $\left[c_1,c_2\right]$ for constants $c_1,c_2\geq 0$ & \qquad $\left[\frac{1}{\log n},\frac{\log\log n}{\log n}\right]$ \\ \midrule \\ [-1em]
        $\Theta(1)$ & logarithmic & \qquad polylog \\ \\ [-1em]
        $\Theta(\log\log n)$ & polylog & \qquad polylog \\ \\ [-1em]
        $\Theta(\log n)$ & polynomial & \qquad polylog \\
	\bottomrule
    \end{tabular}
    \caption{The \leo{growth-rate of} the number of sites sufficient to correctly reconstruct a binary semi-directed level-1 network $N$ under the CFN model with high probability as a function of the number of taxa $n$ (see \cref{thm::CFN} \leo{for a precise statement}). %The structural parameters $\text{depth}(\N)$ and $lc(\N)$ are defined in \cref{sec:preliminaries}. In short, 
    \leo{This growth-rate depends on the maximum of the network depth and the maximum cycle length~$lc(\N)$ (first column, see \cref{sec:preliminaries} for definitions) and on the bounds on the mutation probabilities. In the second column these bounds are constants while in the third column they are functions of~$n$. Inheritance probabilities are bounded by a function of the mutation probabilities (see \cref{prop::successProb}).}
    %$\text{depth}(\N)$ is an upper bound across all internal vertices $v$ and displayed trees $T$ of $\N$ on the shortest topological distance between $v$ and a leaf of the subtree of $T$ induced by $v$. The parameter $lc(\N)$ denotes the maximum cycle length in $\N$. The CFN model is governed by mutation probabilities $p(e)$, $e\in E(\N)$ and the inheritance probabilities are fixed and bounded by a function of the mutation probabilities (see \cref{prop::successProb}). The second and third column depict the result for $p(e)\in [c_1,c_2]$, $e\in E(\N)$, and $p(e)\in [1/\log(n),\log\log(n)/\log(n)]$, $e\in E(\N)$, respectively.}
    \label{tab::2}}
\end{table*}

Our main result is that the introduced method reconstructs the correct network with high probability, under the \CorrB{CFN model} of evolution, for input sequences whose required length scales logarithmically, polynomially, or polylogarithmically with the number of taxa, depending on the parameter regime, see Table~\ref{tab::2}.%\todo{NH: I find the last sentence of the caption a little contradictory. If the probs are fixed, aren't they always bounded by a function of the mutation probs?}

A central concept in our approach is the notion of a \emph{quartet profile}: the set of quartet trees displayed by a network for a given set of four taxa, see Figure~\ref{fig:preliminaries}. %These
Such a quartet profile may either contain one quartet tree (representing tree-like evolution), or multiple quartet trees (representing a cycle, and thus network-like evolution). Our method differentiates between the two by applying the four-point method~\cite{pereira1969note}---a key characteristic in all absolute fast converging methods for tree inference~\cite{warnow2001absolute}---to estimated pairwise distances. By precisely characterizing a strategically chosen subset of these quartet profiles that capture the structure of the network, we can significantly reduce the number of quartet profiles that must be inferred directly from the biological data. Then, we apply newly developed inference rules to deduce the remaining set of quartet profiles, which enables us to leverage an existing reconstruction algorithm~\cite{frohn2024reconstructing} to reconstruct the full network. Note also that quartet profiles have been shown to be identifiable under several models, with proofs either relying on algebraic invariants~\cite{englander2025identifiability} or on concordance factors (a type of summary statistic)~\cite{rhodes2025identifying}.

The inference rules we develop are also interesting in their own right, as they extend and generalize classical results from tree-based phylogenetics. While inference rules for induced quartet trees of phylogenetic trees have been extensively studied, tracing back to~\cite{colonius1981tree} (see also \cite{semple2003phylogenetics} and references therein), their use in a network setting has been very limited so far. %On the one hand,
% Leo: there does not seem to be a tegenstelling here
The \textsc{TiNNIK} software tool~\cite{allman2023tree,allman2024tinnik} for inferring the \emph{tree-of-blobs}---a tree capturing the global branching structure of a \niels{semi-directed} network---\niels{relies on a %inference 
rule to propagate quartet classifications by marking additional quartets as belonging to reticulate regions.}
%for quartets of a \niels{semi-directed} network.
%. On the other hand,
\niels{However, this rule is not an inference rule in the sense used here.} % but rather a rule to propagate quartet classifications by marking additional quartets as belonging to reticulate regions.}
In addition, \cite{huber2018quarnet} introduced inference rules for \emph{quarnets} (4-leaf subnetworks) of undirected level-1 networks (see also~\cite{keijsper2014reconstructing}), but these do not directly translate to our semi-directed setting. In our work, we define and study the closure of a set of quartet profiles, building on well-established ideas from the tree literature. This also leads to a novel quarnet encoding result, showing that any semi-directed level-1 network is encoded by some linear sized set of its induced quarnets (see \cref{cor:linear_encoding}). Although our focus is on level-1 networks, the validity of most of our inference rules is shown for networks of any level, opening up broader avenues for future research.

The structure of the paper is as follows. In \cref{sec:preliminaries} we cover preliminaries on phylogenetic networks and the CFN model. In \cref{sec:dyadic_closure} we introduce novel quartet profile inference rules and use these to define the closure operation. These techniques are used in \cref{sec:DCM}, where we present our main distance-based inference algorithm, which is analyzed under the CFN model in \cref{sec:DCM_analysis}. We end with a discussion in \cref{sec:discussion}.

\begin{figure}
    \centering
    \includegraphics{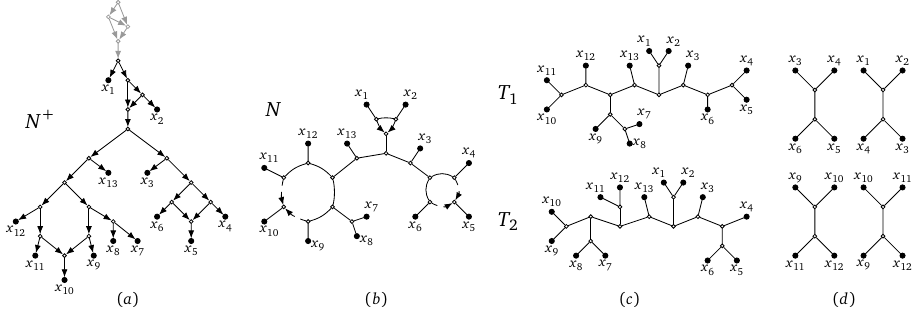}
    \caption{$(a)$: A rooted phylogenetic network~$\N^+$ on $\X = \{x_1, \ldots, x_{13} \}$. \CorrB{The gray vertices are above the LSA of the network. Hence, removing the gray vertices and edges from $\N^+$ gives an LSA network.} %when the gray vertices and edges are not considered.} 
    $(b)$: The semi-directed phylogenetic level-1 network~$\N$ on $\X$ obtained from \CorrB{this LSA network}. 
    %\todo{L: rephrased slightly. N: looks good.}
    %\N^+$ \CorrB{(by excluding the gray vertices and edges)}. 
    $(c)$: Two unrooted phylogenetic trees $T_1$ and $T_2$ on $\X$ displayed by $\N$. %The two trees also form a displayed tree cover of~$\N$. 
    $(d)$: Four quartets displayed by $N$. The quartet $x_3 x_4 | x_5 x_6$ does not form a quartet profile of~$N$, since $N$ also displays another quartet ($x_3 x_6 | x_4 x_5$, not depicted in the figure) on the same four leaves. The quartet $x_1 x_2 | x_3 x_4$ forms a trivial quartet profile, and the two remaining quartets together form the quartet profile $(x_9 x_{10} x_{11} x_{12})$.}
    \label{fig:preliminaries}
\end{figure}

\section{Preliminaries}\label{sec:preliminaries}
\subsection{Phylogenetic Networks}\label{subsec:phylo}

\CorrB{For the sake of generality, we allow phylogenetic networks to contain parallel edges, which can result from a variety of biological processes~\cite{justison2024}; see also, e.g.,~\cite{allman2023tree,banos2019identifying,rhodes2025identifying}. Formally, we call a graph that may contain parallel edges a \emph{multi-graph}; it is \emph{directed} if all edges are directed, and \emph{mixed} if it %contains
may contain both directed and undirected edges.}
%\todo{L: changed the def so trees are also semi-directed. N: agreed.}

\begin{definition}[Rooted network]\label{def:rooted_network}A \emph{(binary) rooted (phylogenetic) network} $N^+$ on a set of at least two taxa $\X$ is a %rooted
directed acyclic multi-graph such that (i) there is a unique vertex with indegree zero (the \emph{root}), which has out-degree two; (ii) each vertex with out-degree zero (a \emph{leaf}) has in-degree one and the set of such vertices is bijectively labeled by $\X$; (iii) all other vertices either have in-degree one and out-degree two (\emph{tree vertices}), or in-degree two and out-degree one (\emph{reticulation vertices}). 
\end{definition}
\CorrB{The \emph{lowest stable ancestor (LSA)} of a rooted network on $\X$ is the lowest vertex through which all directed paths from the root to any leaf in~$\X$ pass. Throughout this work, we assume that all rooted networks on $\X$ are \emph{LSA networks}: networks in which the root is the LSA %of $\X$
%\todo{L: removed ``of $X$'' because we define the LSA of a network, not of a set of taxa}
(see \cref{fig:preliminaries}a). This is a common assumption, since within standard modeling frameworks, the structure above the LSA can not be resolved from quartet data~\cite{banos2019identifying}.}
The two edges directed towards a reticulation vertex are called \emph{reticulation edges}.

\begin{definition}[Semi-directed network]\label{def:semi_directed}
A \emph{(binary) semi-directed (phylogenetic) network} $\N$ on at least two %leaves 
taxa $\X$ is a mixed multi-graph that can be obtained from a binary rooted phylogenetic LSA
% (LSA)
network~$N^+$ on~$\X$ by undirecting all non-reticulation edges and suppressing the former root.
\end{definition}

%Note that in general, a semi-directed network may have parallel edges.
The reticulation edges of a semi-directed network are precisely its directed edges, and the reticulation vertices are the vertices with in-degree two.
\niels{A leaf~$x$ is \emph{below} a reticulation vertex~$r$ if there is a path from $r$ to $x$ consisting of no edges oriented towards~$r$.}
%Clearly, we can define reticulation vertices and reticulation edges in semi-directed networks analogously as for rooted networks.\todo{then why not do it?}
% Leo: moved the def. of blob here because we don't use it before
\CorrB{A \emph{cut-edge} is an edge whose removal separates the graph into two connected components and the cut-edge is \emph{nontrivial} if the two components contain at least two vertices each.}
A \emph{blob} of any multi-graph %(and hence also of a rooted network)
(including directed and mixed multi-graphs)
is a maximal subgraph without any cut-edges. A blob is an \emph{$m$-blob} if it has $m$ \niels{cut-}edges incident to it \CorrB{and it is \emph{internal} if it does not contain a leaf.}
%This allows us to define
A (rooted or semi-directed) network is \emph{level-$k$}, for some \niels{integer}~$k\geq 0$, if there are at most~$k$ reticulation vertices in each blob of the network.
%, for some non-negative integer~$k$.
In this work, we mostly focus on level-1 networks, \CorrB{which consist of only tree-like parts and vertex disjoint cycles since we assume our networks to be binary.} %Another property a (rooted or semi-directed) network may have is being \emph{outer-labeled planar}. A network has this property if there exists a planar representation of the network in which all leaves lie in the unbounded face. 
See \cref{fig:preliminaries}b for an example of a semi-directed level-1 network. \niels{We denote the maximum topological length of a cycle in $\N$ by $lc(\N)$.} A semi-directed level-1 network with a single \CorrB{internal} blob \CorrB{and which contains exactly one reticulation vertex} is called a \emph{sunlet (network)}.

A level-0 network is commonly known as a phylogenetic tree. We distinguish between a \emph{rooted phylogenetic tree} (a level-0 rooted network) and an \emph{unrooted phylogenetic tree} (a level-0 semi-directed network). An unrooted phylogenetic tree on at least two leaves $\X$ is a \emph{displayed tree} of a semi-directed network $\N$ on $\X$ if the tree can be obtained from the network by deleting one reticulation edge per reticulation vertex, undirecting the remaining reticulation edges, exhaustively deleting leaves not in $\X$, and suppressing degree-2 vertices. %A pair of displayed trees $(\T_1, \T_2)$ is a \emph{tree cover} of a semi-directed network~$\N$ if, for every reticulation vertex~$v$ in~$\N$, the tree~$\T_1$ was obtained by deleting one reticulation edge entering~$v$ and $\T_2$ was obtained by deleting the other reticulation edge entering~$v$. 
See \cref{fig:preliminaries}c for an example.

%If we do not remove the leaves not in $\X$ and also not suppress the degree-2 vertices, we instead call such a tree an \emph{unreduced displayed tree}.
%The \emph{tree-of-blobs} $\T ( \N )$ of a semi-directed network $\N$ is obtained by contracting every blob $\B$ to a single vertex $v$, in which case we say that \emph{$v$ represents $\B$}. By definition, the tree-of-blobs $\T (\N)$ must be an undirected phylogenetic tree on $\X$ since a semi-directed network contains no 2-blobs.

\smallskip

In \cref{def:restriction} we formalize the notion of a \emph{subnetwork} (or \emph{subtree} for trees). To this end, we let an \emph{up-down path} between two leaves $x_1$ and $x_2$ of a semi-directed network be a path of $k$ edges where the first $\ell$ edges are undirected or directed towards $x_1$ and the last $k - \ell$ edges are undirected or directed towards $x_2$. %Here, we consider undirected edges to be bidirected.
\begin{definition}[Subnetwork]\label{def:restriction}
Given a semi-directed network $\N$ on $\X$ and some $\Y \subseteq \X$ with $| \Y | \geq 2$, the \emph{subnetwork} of $\N$ induced by $\Y$ is the semi-directed network $\N|_\Y$ obtained from $\N$ by taking the union of all up-down paths between leaves in $\Y$, followed by exhaustively suppressing degree-2 vertices.
\end{definition}

Given a semi-directed network $\N$ on $\X$, a 4-leaf subtree of a displayed tree of $\N$ is called a \emph{displayed quartet}.\footnote{\CorrB{In \cref{prop:displayed_subnet_commutative} of \cref{sec:commutative} we prove that taking displayed trees and inducing subnetworks are commutative. Hence, equivalently, a displayed quartet can also be defined as a displayed tree of a 4-leaf subnetwork.}} A displayed quartet on $\{a,b,c,d\}$ has four trivial cut-edges: one cutting off each leaf. Next to that, the quartet has exactly one nontrivial cut-edge inducing a nontrivial partition of its leaves ($ab|cd$, $ac|bd$, or $ad|bc$). In general, given a semi-directed network $\N$ on $\X$ and a partition $A|B$ of $\X$ (with $A$ and $B$ both non-empty), we say that $A|B$ is a \emph{split} of $\N$ if there exists a cut-edge of $\N$ whose removal disconnects the leaves in $A$ from those in~$B$. The split and the cut-edge are \emph{nontrivial} if the corresponding partition is nontrivial, that is, if $|A|, |B| \geq 2$. As already used above, for splits with few leaves we sometimes omit the set notation, e.g., we use $ab | cd$ instead of $\{a, b\} | \{c, d\}$.

\smallskip
The main data type used in this work is the following. See \cref{fig:preliminaries}d for an example.
\begin{definition}[Quartet profile]
A \emph{quartet profile} on four leaves $\{a,b,c,d\}\subseteq\X$ is a set of quartet trees on $\{a,b,c,d\}$. Given a semi-directed network $\N$ on $\X$, the quartet profile of~$\N$ on $\{a,b,c,d\} \subseteq \X$ is the set of displayed quartets of~$\N$ on $\{a,b,c,d\}$.
\end{definition}
If the quartet profile on $\{a,b,c,d\}$ contains a single quartet tree, we denote it by the split induced by this tree, say $ab|cd$, and call it a \emph{trivial} quartet profile. \CorrB{Note that this implies that the quartet $ab|cd$ is an induced 4-leaf subtree of every displayed tree of the network.} If instead the quartet profile contains two quartet trees, we denote the quartet profile by the circular ordering that is congruent with the \CorrB{nontrivial splits of the two trees, i.e., in which each split corresponds to opposite pairs in the ordering}. For instance, if the quartet profile on $\{a,b,c,d\}$ contains exactly the trees with splits $ab|cd$ and $ad|bc$, then the quartet profile will be denoted by $(abcd)$ (or equivalently, e.g., $(bcda)$ or $(dcba)$). \CorrB{See also \cref{fig:preliminaries}d, where the quartet trees $x_9 x_{10} | x_{11} x_{12}$ and $x_9 x_{12} | x_{10} x_{11}$ form a quartet profile and both are congruent with the circular ordering $(x_9 x_{10} x_{11} x_{12})$, which also matches the ordering of the leaves around the left-most cycle of the network~$N$ in \cref{fig:preliminaries}b.} The leaf set of a quartet profile $q$ is denoted by $\LL (q)$ and the set of all quartet profiles of a network $\N$ is denoted by $\Q (\N)$. In general, a \emph{set of quartet profiles $\Q$ on $\X$} is a set such that $\X = \bigcup_{q\in \Q} \LL (q)$ and there is at most one quartet profile per 4-leaf subset of $\X$.

It follows from \cite{rhodes2025identifying} that the quartet profiles (referred to there as `4-taxon circular order information') of any \emph{outer-labeled planar} network (which include the level-1 networks we consider \CorrB{\cite{moulton2022planar}}) never contain all three possible quartet trees. We emphasize that knowing whether a quartet profile is displayed by a network amounts to knowing \emph{all} displayed quartets for the considered set of four leaves. \CorrB{Hence, a network has only trivial quartet profiles if and only if all its displayed trees induce the same set of quartets. Since quartets uniquely encode a tree~\cite{semple2003phylogenetics}, this is equivalent to the network having a unique displayed tree.}

\smallskip 

We end with a useful result specifying exactly which features of a network are uniquely determined and can be reconstructed from quartet profiles. Throughout this manuscript, we will use the notation~$\NN$ introduced in this proposition, which represents the main object we aim to construct. For convenience, we may occasionally refer to~$\NN$ as a ``network'', although it is not one in the strict sense of the term. By \emph{suppressing} a $k$-cycle in a semi-directed network, we mean replacing the cycle by a single vertex and suppressing resulting degree-2 vertices. For convenience, we refer to the operation of undirecting reticulation edges as \emph{suppressing reticulation vertices}.

\begin{proposition}[\cite{banos2019identifying, frohn2024reconstructing}]\label{prop:qp_reconstruct}
   Let $\N$ be a semi-directed level-1 network on $n$ leaves and denote by $\NN$ the mixed graph obtained from $\N$ by suppressing 2-cycles, 3-cycles and reticulation vertices in 4-cycles. Then, $\NN$ is uniquely characterized by $\Q (\N)$ and can be reconstructed from $\Q (\N)$ in $\bigO (n^2)$ time.
\end{proposition}

%The \emph{support} of a set $\Q \subseteq \Q (\N)$ of quartet profiles, denoted by $\supp (\Q)$, is the set of size four leaf sets of the quartet profiles in~$\Q$. 

% In the following proposition we show the intuitive fact that, equivalently, the displayed quartets of $\N$ are the quartets induced by the displayed trees of $\N$.
% \begin{proposition}
%     Let $\N$ be a semi-directed network on $\X$ and let $T$ be a binary quartet tree. Then, $T$ is a displayed quartet of $\N$ if and only if $T$ is a subtree of a displayed tree of $\N$.
% \end{proposition}
% \begin{proof}
% Without loss of generality, suppose that $T$ is the quartet tree with split $ab|cd$. For the forward direction, note that $T$ can be obtained from $\N|_{\{a,b,c,d\}}$ by deleting reticulation edges and then applying the subnetwork operations again. These reticulation edges come from reticulation edges in $\N$. Then, $T$ is also a subtree of a displayed tree of $\N$ where we delete these reticulation edges. For the other direction, note that if $T$ is a subtree of a displayed tree of $\N$. There will be up-down paths in $\N$ whose union has the split $ab|cd$. Then, ... [Perhaps formulate the definition the other way around actually.]
% \end{proof}

\subsection{The CFN model and distances in phylogenetic networks}
Given a rooted phylogenetic network $\N=(V,E)$, for each $v\in V$ we associate a random variable $Y(v)$ with a binary state space, corresponding to purine or pyrimidine bases in DNA sequences. For each edge $e=(u,v)\in E$ that is not a reticulation edge, we define $p(e)$ as the probability that the states of $Y(u)$ and $Y(v)$ differ and call it the \emph{mutation probability} of $e$. 
We \niels{make the standard exclusion $p(e)\neq 0$ and $p(e)\neq 1/2$} because such a probability would mean that either there is no change on $e$ or variables $Y(v)$ and $Y(u)$ are random with respect to each other~\cite{erdos1999logs1}. Moreover, biologically, nucleotide states are more likely to remain unchanged than to mutate. Therefore, $p(e)\in(0,1/2)$.
%If $N$ is a phylogenetic tree, then additionally for an edge $e=(u,v)$ with $p(e)>1/2$ we observe that relabeling the state of $v$ yields mutation probability $p'(e)=1-p(e)<1/2$. Subsequently, one can recursively relabel the state of the children of $v$ to keep the probability of each site pattern of the network unchanged and therefore restrict mutation probabilities without loss of generality to values in $(0,1/2)$. However, if $N$ contains cycles, then this relabeling procedure cannot work for reticulation vertices. Thus, we require $p(e)\in(0,1/2)\cup(1/2,1)$. 
Observe that probabilities $p(e)$ yield a Markov model on the displayed trees of $\N$, called the CFN model~\cite{neyman1971molecular}, which is an instance of the general Markov model for transition matrices
\begin{align*}
    M^e=\left(\begin{array}{cc}
    1-p(e) & p(e)\\
    p(e) & 1-p(e)
    \end{array}\right)
\end{align*}
and a uniform root distribution. For reticulation edges $e=(u,v)\in E$ we also define the \emph{inheritance probability} $\gamma(e)\in(0,1)$ as the proportion of \niels{sites in the alignment} of $v$ contributed by $u$. \CorrB{The following result is well-known in the literature (e.g. page 13 in~\cite{hendy1993spectral}).} %We include a proof for completion.
\begin{lemma}\label{lem::treePath}
    For the path in a phylogenetic tree constituted by edges $e_1,\dots,e_m$ and with endpoints $v_1$ and $v_2$, we have
    \begin{align*}
        \mathbb{P}[Y(v_1)\neq Y(v_2)]=\frac{1}{2}\left(1-\prod_{i=1}^m(1-2p(e_i))\right)
    \end{align*}
    under the CFN model.
\end{lemma}
%\begin{proof}
%By induction on $m$: for $m=1$ our claim holds because $p(e_1)=(1-(1-2p(e_1)))/2$. Assume our induction hypothesis holds for $m-1$. Let the path constituted by edges $e_1,\dots,e_{m-1}$ have endpoints $v_1$ and $v_3$. Hence, by induction
%\begin{align*}
%    \mathbb{P}[Y(v_1)=Y(v_2)]&=\mathbb{P}[Y(v_1)=Y(v_3)]\cdot\mathbb{P}[Y(v_3)=Y(v_2)]=(1-2\cdot\mathbb{P}[Y(v_1)\neq Y(v_3)])\cdot(1-2p(e_m))\\
%    &=\prod_{i=1}^{m}(1-2p(e_i)).
%\end{align*}
%\end{proof}
\noindent Observe that \cref{lem::treePath} can be applied to every displayed tree $\T$ of $\N$. Hence, we can encode the magnitude of the mutation probabilities in the CFN model as edge lengths of $\N$ \CorrB{as follows:} for each displayed tree $\T$ of $\N$, we define a $n\times n$ matrix $D(\T)$ of pairwise distances between leaves $\mathcal{X}=\{x_1,\dots,x_n\}$ such that $D(\T)$ \CorrB{are additive distances} on $\T$. Specifically, for the unique path in $\T$ constituted by edges $e_1,\dots,e_m$ and with endpoints $x_i,x_j\in\mathcal{X}$, we define
\begin{align*}
    D(\T)_{ij}=-\frac{1}{2}\log\left(1-2\cdot\mathbb{P}[Y(x_i)\neq Y(x_j)]\right)=-\sum_{l=1}^m\frac{1}{2}\log\left(1-2p(e_l)\right)
\end{align*}
to adjust the distance matrix to our choice of mutation probabilities under the CFN model. \CorrB{This means, $D(\T)_{ij}$ is the sum of positive edge lengths $\ell(e_j)$, $j\in\{1,\dots,m\}$, defined in general by
\begin{align*}
    \ell(e)&=-\frac{1}{2}\log\theta(e)~&~&\forall\,e\in E,\,\theta(e)=1-2p(e).
\end{align*}
In other words, $D(\T)$ defines \CorrB{additive distances} on $\T$.} Notice that this definition of $D(\T)$ coincides with the paralinear distance~\cite{lake1994reconstructing} and the LogDet distance~\cite{steel1994recovering} for $\T$ under the CFN model.

Now, let $\mathcal{S}$ be a sequence alignment on $\mathcal{X}$ of length $k$ generated under the CFN model. \CorrB{Hence, $\mathcal{S}$ is a collection of binary vectors $s\in\{0,1\}^k$.} We let $d_H(i,j)$ denote the \emph{Hamming distance} between the sequences $s^i,s^j\in \mathcal{S}$ to define the \emph{dissimilarity score} of sequences $s^i$ and $s^j$ as $h^{ij}=d_H(i,j)/k$. Then, the expectation $\mathbb{E}[h^{ij}]$ can be seen as the probability $\mathbb{P}[Y(x_i)\neq Y(x_j)]$ between sequences $s^i$ and $s^j$. Therefore, we define the \emph{empirical distance} of $s^i$ and $s^j$ by
\begin{align*}
d_{ij}=-\frac{1}{2}\log\left(1-2h^{ij}\right)
\end{align*}
to ensure that $d_{ij}$ converges to $D(\T)_{ij}$ in probability as $k$ grows for any fixed displayed tree $\T$ of $\N$~\cite{steel1994recovering}. This means, for fixed inheritance probabilities $\gamma(e)$, any reticulation vertex $v$ in a rooted binary phylogenetic network $N$ and displayed trees $T_1$ and $T_2$ of $N$ which differ only in distinct reticulation edges $e_1$ and $e_2$ incident to $v$, respectively, we know that $d_{ij}$ converges to $\gamma(e_1)D(\T_1)_{ij}+\gamma(e_2)D(\T_2)_{ij}$ in probability as $k$ grows.

Using the empirical distances as input, \cite{erdos1999logs1} provided an absolute fast converging method to reconstruct the tree~$\T$ with probability $1-o(1)$ if
\begin{align}\label{k::LB}
    k>\frac{c\log n}{\left(1-\max\limits_{e\in E}\theta(e)\right)^2\left(\min\limits_{e\in E}\theta(e)\right)^{\text{depth}(\T)}}
\end{align}
for a constant $c$ and $$\text{depth}(T)=\max_{e=(v_1,v_2)\in E(T)}\max_{i\in\{1,2\}}\{\text{shortest topological distance from $v_i$ to a leaf after removing $e$}\}.$$
\begin{table*}[!t]
\centering
    \begin{tabular}{l|ll}
     \toprule
        & Range of mutation probabilities $p(e)$, $e\in E(T)$ & \\
     	depth$(T)$ & $\left[c_1,c_2\right]$ for constants $c_1,c_2\geq 0$ & $\left[\frac{1}{\log n},\frac{\log\log n}{\log n}\right]$ \\ \midrule \\ [-1em]
        $\Theta(1)$ & logarithmic & polylog \\ \\ [-1em]
        $\Theta(\log\log n)$ & polylog & polylog \\ \\ [-1em]
        $\Theta(\log n)$ & polynomial & polylog \\ \\ [-1em]
	\bottomrule
    \end{tabular}
    \caption{\leo{Asymptotic behaviour} of the righthand side of bound~\eqref{k::LB} for different \leo{growth-rates for} $\text{depth}(T)$ and different ranges for mutation probabilities $p(e)$, $e\in E(T)$, with~$n$ the number of taxa.}
    \label{tab::1}
\end{table*}
As a shorthand notation we denote $\text{depth}(N)$ as the largest chosen number $\text{depth}(T)$ among all displayed trees $T$ of $N$. In other words, $\text{depth}(N)$ is an upper bound across all internal vertices $v$ and displayed trees $T$ of $N$ on the shortest topological distance between $v$ and a leaf of the subtree of $T$ induced by $v$. Table~\ref{tab::1} shows different values for the bound~\eqref{k::LB} on the number of sites $k$. Notice that $\text{depth}(T)=\Theta(1)$ and $\text{depth}(T)=\Theta(\log n)$ occur for the completely unbalanced and completely balanced binary tree $T$, respectively. \niels{Additionally,} $\text{depth}(\T)=\mathcal{O}(\log\log n)$ typically occurs for Yule trees~\cite{erdos1999logs1}. Furthermore, if values $\theta(e)$ are restricted to be either large or small, then the bound on $k$ becomes logarithmic or polynomial, respectively~\cite{mossel2004phase}. 

\CorrB{We call the length of a shortest path between two vertices $v_i,v_j$ in a tree $T$ the \emph{graphical distance} between $v_i$ and $v_j$ and denote it by $D_{ij}$. The matrix of pairwise graphical distances between vertices in $T$ is denoted by~$D$. Here, the length of a path is either the number of edges of a path if $T$ is unweighted or the sum of the edge lengths of the path if $T$ is weighted. Then, we will make use of the following characterization of additive distances~\cite{buneman1974note}:}
\begin{lemma}[Four-point condition]
    A graph is a tree if and only if it is connected, \niels{3-cycle}-free and has graphical distance $D$ satisfying the four-point condition:
    \begin{align*}
        D_{ij}+D_{kl}\leq\max\left\{D_{ik}+D_{jl},D_{il}+D_{jk}\right\}
    \end{align*}
    for all vertices $v_i,v_j,v_k,v_l$ in the graph. Equivalently, the four-point condition requires that exactly one of the following properties holds:
    \begin{align*}
        D_{ij}+D_{kl}&< D_{ik}+D_{jl}=D_{il}+D_{jk},\\
        D_{ik}+D_{jl}&< D_{ij}+D_{kl}=D_{il}+D_{jk},\\
        D_{il}+D_{jk}&< D_{ik}+D_{jl}=D_{ij}+D_{kl}
    \end{align*}
    for all vertices $v_i,v_j,v_k,v_l$ in the graph.
\end{lemma}

\section{Inference rules and dyadic closure}\label{sec:dyadic_closure}

In this section we develop inference rules for quartet profiles and utilize them to characterize a small set of quartet profiles that suffice to reconstruct the whole network. Although our primary focus is on level-1 networks, we prove that several of the inference rules hold more generally for networks of any level, which could open further directions for research.

\subsection{Inference rules}

Inference rules for quartets of unrooted phylogenetic trees are rules of the type 
$ab|cd + ab|ce \mapsto ab|de$,
meaning that if $ab|cd$ and $ab|ce$ are quartets induced by a tree, this implies that $ab|de$ is a quartet induced by the same tree. Such inference rules have been extensively studied for phylogenetic trees \cite[e.g.][]{bandelt1986reconstructing,erdos1999logs1}, since they allow to extract more quartets from only limited quartet information available. More recently, inference rules for \emph{quarnets} (4-leaf subnetworks) of fully undirected networks have also been studied \cite{huber2018quarnet}. The method presented in this paper partially relies on new inference rules for quartet profiles. We note that since we study semi-directed networks and hence the subnetwork induced by fewer leaves might force the loss of information on reticulation events, these rules are inherently different from the undirected quarnet inference rules presented in \cite{huber2018quarnet}. 

\smallskip

Formally, an \emph{inference rule} for quartet profiles is a rule of the type $$q_1 + \cdots + q_k \mapsto q_1' + \cdots + q_l'$$%\mf{Since $x_i$ are taxa in other sections better use $q_i$ as a different set of symbols}\nh{Agreed}
that can be applied to a set of distinct quartet profiles $\Q=\{q_1, \ldots, q_k\}$ and \CorrB{yields} a set of new quartet profiles $\Q'=\{q_1, \ldots, q_l\}$. \CorrB{Recall that each $q_i$ forms a set of quartet trees: if $|q_i| = 1$, it is identified with the split of its unique quartet tree, e.g.\ $ab | cd$, and if $|q_i| = 2$, it is represented using circular ordering notation, e.g.\ $(abcd)$.} We say that an inference rule is \emph{valid} if the following implication holds for all semi-directed networks $\N$ (possibly restricted to a certain subclass of semi-directed networks)
\begin{align*}
\Q\subseteq\Q(N)~~\Rightarrow~~\Q'\subseteq\Q(N).
\end{align*}
The \emph{order} of an inference rule is its input size~$k$. A second order inference rule is called \emph{dyadic}, while a third order rule is \emph{triadic}.

We first introduce two well-known dyadic inference rules for quartets of unrooted phylogenetic trees (see e.g. \cite{bandelt1986reconstructing,erdos1999logs1}):

%$$R2: ab|cd + ac|de = abc|de + ab|cde$$
%    ab|de + ac|df + bc|ef \mapsto ab|df + ab|ef + ac|de + ac|ef + bc|de + bc|df
%$$R3: ab|de + ac|df + bc|ef = abc|def$$
%The third rule is a \emph{triadic} (i.e. third order) inference rule for phylogenetic trees (see \cite{grunewald2005quartet}).%Similarly, rule \ref{eq:R3} can be formulated as implying that the subnetwork induced by $\{a,b,c,d,e, f\}$ has the split $abc|def$.

\begin{equation}\label{eq:R1}
    ab|cd + ac|de \mapsto ab|de + bc|de;\tag{R1}
\end{equation}
%$$R1: ab|cd + ab|ce = ab|cde$$
\begin{equation}\label{eq:R2}
    ab|cd + ab|ce \mapsto ab|de .\tag{R2}
\end{equation}
In the following proposition we prove that any rule that is valid for quartets of an unrooted phylogenetic tree---which thus includes inference rules \ref{eq:R1} and \ref{eq:R2}---is valid for semi-directed networks, when considering (necessarily trivial) quartet profiles. We note in passing that this also means that these are valid inference rules for quarnet-splits (as defined in \cite{frohn2024reconstructing}).
\begin{proposition}\label{prop:rule1_valid}
    Any valid quartet inference rule for unrooted phylogenetic trees is a valid quartet profile inference rule for semi-directed networks.
\end{proposition}
\begin{proof}
Denote by~$R$ a quartet inference rule $q_1 + \cdots + q_k \mapsto q_1' + \cdots + q_l'$ (with $k,l \geq 1$) that is valid for any unrooted phylogenetic tree. Let $\N$ be a semi-directed network and assume that all $q_i$ are also trivial quartet profiles of $\N$ (i.e. there is only one displayed quartet tree for each of the corresponding four-leaf subsets). Next, let $\T$ be a displayed tree of $\N$ \CorrB{on the same leaf set}. Then, since all quartet profiles $q_i$ are trivial, all of them are also induced quartets of the tree $\T$. Hence, by the valid rule~$R$, all the $q_j'$ are induced quartets of $\T$. Since $\T$ was an arbitrary displayed tree of $\N$ \CorrB{on the same leaf set}, all $q_j'$ are quartets of any of the displayed trees of $\N$, and hence they are trivial quartet profiles of~$\N$.
\end{proof}

%Equivalently, the two rules \ref{eq:R1} and \ref{eq:R2} imply that the subnetwork induced by $\{a,b,c,d,e\}$ has the splits $ab|cde$ (for rule \ref{eq:R1}) or the splits $ab|cde$ and $abc|de$ (for rule \ref{eq:R2}). 

Building upon the two inference rules~\ref{eq:R1} and~\ref{eq:R2}, we introduce two novel inference rules (the first one being dyadic, and the second one triadic) that also consider nontrivial quartet profiles:
\begin{equation}\label{eq:R3}
    (abcd) + ac|de \mapsto ab|de + bc|de ;\tag{R3}
\end{equation}
\begin{equation}\label{eq:R4*}
    (abcd) + ab|de + bc|de \mapsto (abce). \tag{R4$^*$}
\end{equation}
In the following lemma we show that these inference rules are valid for any semi-directed network.
\begin{lemma}\label{lem:rule3_valid}
    The quartet profile inference rules \ref{eq:R3} and \ref{eq:R4*} are valid for any semi-directed network.
\end{lemma}
\begin{proof}
We first consider rule \ref{eq:R3} and assume that $(abcd)$ and $ac|de$ are quartet profiles of $\N$.  Then, we can partition the displayed trees of $\N$ into the set $\mathcal{T}_1$ with quartets $ab|cd$, and the set $\mathcal{T}_2$ with quartets $ad|bc$. Since $ac|de$ is a trivial quartet profile of $\N$, it is a quartet of all of these displayed trees. We can now apply the valid quartet inference rule \ref{eq:R1} to the quartets in~$\mathcal{T}_1$, and obtain that $ab|de$ and $bc|de$ are quartets of those trees. After switching the roles of $a$ and $c$ we can rewrite~\ref{eq:R1} to $bc|ad + ac|de \mapsto ab|de + bc|de$, and we obtain that $ab|de$ and $bc|de$ are also quartets of the trees in $\mathcal{T}_2$. Since this covers all displayed trees, these two quartets are (trivial) quartet profiles of~$\N$.

For rule \ref{eq:R4*}, suppose that $(abcd)$, $ab|de$ and $bc|de$ are quartet profiles of~$\N$. We again have the two sets of displayed trees,
$\mathcal{T}_1$ with quartets $ab|cd$, and $\mathcal{T}_2$ with quartets $ad|bc$.
% $\mathcal{T}_1$ and $\mathcal{T}_2$\todo{again note what these sets are}. 
Furthermore, in any of these displayed trees, $ab|de$ and $bc|de$ are quartets. Now, note that we can rewrite rule \ref{eq:R2} to $ab|cd +ab|de \mapsto ab|ce$ by switching the roles of $c$ and $d$. Applying this rule to the trees in $\mathcal{T}_1$, gives us that $ab|ce$ are quartets in those trees. Similarly, we can also rewrite rule \ref{eq:R2} to $bc|da + bc|de \mapsto bc|ae$. Applying this rule to the trees in $\mathcal{T}_2$, gives us that $bc|ae$ are quartets in those trees. Since we have considered all displayed trees, this means that $(abce)$ is a quartet profile of~$\N$.
\end{proof}
In case the semi-directed network $N$ is level-1, inference rule \ref{eq:R4*} can be simplified to a dyadic inference rule:
\begin{equation}\label{eq:R4}
    (abcd) + bc|de \mapsto (abce).\tag{R4}
\end{equation}
This is useful, since exhaustively applying dyadic inference rules to a set of quartets requires less time than applying higher order inference rules. The following lemma shows validity.

\begin{lemma}\label{lem:rule4_valid}
     The quartet profile inference rule \ref{eq:R4} is valid for any semi-directed level-1 network.
\end{lemma}
\begin{proof}
Suppose that $(abcd)$ and $bc|de$ are quartet profiles of a semi-directed level-1 network $\N$. It suffices to consider the case where $\N$ has exactly five leaves, otherwise we can consider the subnetwork induced by~$\{a,b,c,d,e\}$. Moreover, we can disregard 2- and 3-cycles, since they do not change the displayed quartets. We now consider two cases, depending on whether $\N$ is a sunlet \niels{or not}.

\textbf{Case 1:} \emph{$\N$ is a sunlet.} Since $bc|de$ is a trivial quartet profile, the subnetwork of $\N$ induced by $\{b,c,d,e\}$ also induces the split $bc|de$. \CorrB{Now note that a 4-leaf subnetwork of a sunlet network induces a nontrivial split if its leaf set excludes the leaf below the reticulation, and it does not induce a nontrivial split otherwise.} Hence, $a$ must be the leaf below the reticulation of the sunlet. Furthermore, the quartet profile $(abcd)$ indicates that the leaves of $\N$ are either ordered as $(a,b,c,d,e)$ or $(a,b,c,e,d)$ around the sunlet (see \cref{fig:inference_rule_level1}a). In both cases, $(abce)$ is a quartet profile of~$\N$.

\textbf{Case 2:} \emph{$\N$ is not a sunlet.} The existence of the quartet profile $(abcd)$ shows that $\N$ \emph{does} have a 4-cycle. Consequently, $\N$ consists of one four-cycle with an incident nontrivial cut-edge. Since $bc|de$ is a quartet profile, this cut-edge must induce the split $abc|de$. Hence, $\N$ is one of the four networks in \cref{fig:inference_rule_level1}b, and it follows that $(abce)$ is a quartet profile.
\end{proof}

\begin{figure}
    \centering
    \includegraphics{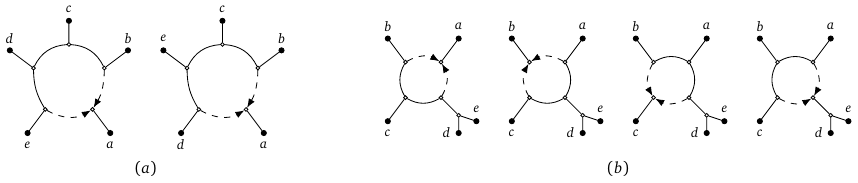}
    \caption{All semi-directed level-1 networks on $\{a,b,c,d,e\}$ with quartet profiles $(abcd)$ and $bc|de$ (disregarding possible 2- and 3-cycles). The networks in subfigure $(a)$ \CorrB{are sunlets that} have a 5-cycle and the networks in subfigure $(b)$ \CorrB{are not sunlets, having} a 4-cycle \CorrB{and a node forming its own internal blob}.}
    \label{fig:inference_rule_level1}
\end{figure}

%(abcd) + (acde) + bc|de \mapsto (abce) + (abde)
%$$ab|cd + ac|de = ab|ce + ab|de + bc|de$$
%$$ab|cd + ab|ce = ab|de$$
%$$ab|de + bc|de + bc|ef = ab|df + ab|ef$$

%Note that the inference rules \ref{eq:R1}-\ref{eq:R4} directly imply inference rules for quarnets of level-1 networks (disregarding the location of possible triangles). For example, the proof of inference rule \ref{eq:R4} can be used to show that $(\underline{a}bcd) + ab|ce \mapsto (\underline{a}bed)$ is a valid inference rule for such quarnets. Here, the underlined letters indicate that $a$ is below the reticulation, whereas $ab|ce$ indicates a quarnet with nontrivial split $ab|ce$.

\subsection{Representative quartet profiles and dyadic closure}

In this section, we present a small set of carefully selected quartet profiles that uniquely determine a semi-directed level-1 network. To do so, we first introduce some additional technical definitions.

\smallskip
Let $\N$ be a semi-directed level-1 network on $\X$ and $B$ an internal $k$-blob. Since $\N$ is level-1, such blobs are either cycles or single degree-3 vertices. Thus, they naturally induce a \emph{circular order} $\C = (X_1, \ldots , X_k)$ of subsets of $\X$ (with $X_1 | \ldots | X_k$ a $k$-partition of $\X$), \CorrB{defined up to reversal and cyclic permutations}. See \cref{fig:dyadic_closure_proof}a for an example. For each of the subsets $X_i$, we call the unique leaf $x_i \in X_i$ that is topologically the closest to $B$ (in terms of the number of edges) its \emph{representative leaf}. To ensure uniqueness, ties are resolved by choosing the leaf with the smallest index. The set of representative leaves $x_i$ corresponding to the circular order $\C$ is denoted by $\RR (\C)$.

Let $uv$ be a nontrivial cut-edge of a semi-directed network $\N$. \CorrB{Then,} $u$ (resp. $v$) is contained in a \CorrB{$(k+1)$-blob} $B_X$ (resp. \CorrB{$(l+1)$-blob} $B_Y$) for some \CorrB{$k+1\geq 3$ (resp. $l+1 \geq 3$)}. Recall that $uv$ induces a nontrivial split $X|Y$ of the leaves of~$\N$. Moreover, $B_X$ (resp. $B_Y$) induces a circular order $
\C_1 = (X_1 , \ldots , X_k , Y)$ (resp. $\C_2 = (X, Y_1 , \ldots , Y_l)$), where $X = \bigcup_i X_i$ (resp. $Y = \bigcup_j Y_j$). Then, we let $\D = (X_1 , \ldots , X_k) | (Y_1 , \ldots , Y_l)$ be the \emph{2-circular order} induced by $uv$, capturing both the induced split and the two induced circular orders. \CorrB{Note that $X_1$ and $X_k$ (resp. $Y_1$ and $Y_l$) are the two neighbors of $Y$ (resp. $X$) in $\C_1$ (resp. $\C_2$). In the 2-circular order $\D$, these are fixed as the outermost elements of the two sequences $(X_1 , \ldots , X_k)$ and $(Y_1 , \ldots , Y_l)$. Hence, unlike circular orders (which are defined up to cyclic permutation and reversal), $\D$ is defined only up to reversal of the two sequences, as it also records where the cut-edge attaches, namely between $X_1$ and $X_k$ and between $Y_1$ and $Y_l$.} See \cref{fig:dyadic_closure_proof}b and c for two examples. The set of representative leaves $\{x_1, \ldots, x_k, y_1, \ldots, y_l\}$ corresponding to $\D$ is again denoted by $\RR (\D)$, where $x_i \in \RR (\C_1)$ and $y_j \in \RR (\C_2)$.

\begin{figure}[htb]
    \centering
    \includegraphics{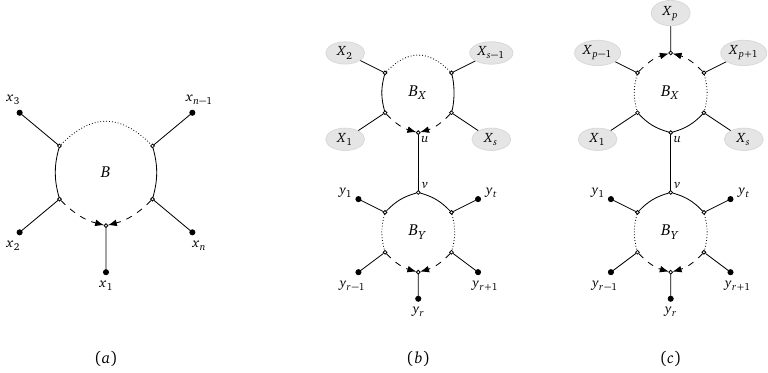}
    \caption{(a): A sunlet network on $\{x_1, \ldots, x_n \}$. The blob~$B$ induces the circular order $(\{x_1\}, \ldots, \{x_n\})$. (b/c): Two types of semi-directed level-1 networks on $X_1 \cup \ldots \cup X_s \cup \{y_1\} \cup \ldots \cup \{y_t\}$, with $s, t \geq 2$, $1 \leq r \leq t$ and $1 \leq p \leq s$ (for subfigure (c)). The $y_j$ represent leaves, and the $X_i$ represent subsets of leaves.
    The networks both have a nontrivial cut-edge $uv$ (with its endpoints in the blobs $B_X$ and $B_Y$) inducing the 2-circular order $(X_1, \ldots, X_s)|(\{y_1\}, \ldots, \{y_t\})$.}
    \label{fig:dyadic_closure_proof}
\end{figure}

\smallskip
We now extend the concept of a representative quartet from \cite{erdos1999logs1}---originally developed for rooted trees---to the context of quartet profiles and semi-directed level-1 networks. 

\begin{definition}[Representative quartet profile]\label{def:repr_quartet}
    Let $\N$ be a semi-directed level-1 network with leaf set $\X$. Then, the set $\R (\N) \subseteq \Q (\N)$ of \emph{representative quartet profiles} is constructed as follows. For each cut-edge with induced 2-circular order $\D = (X_1 , \ldots , X_k ) | (Y_1 , \ldots , Y_l)$ of $\X$ and $\RR (\D) = \{x_1, \ldots, x_k, y_1, \ldots, y_l\}$, 
    \begin{enumerate}[label=(\roman*)]
    \item $x_1 x_k | y_1 y_l \in \R (\N)$.
    \end{enumerate}
    For each $k$-blob $B$ ($k\geq 4$) with induced circular order $\C = (X_1 , \ldots , X_k)$ of $\X$ such that $X_1$ is the set of leaves below the reticulation \niels{in $B$} and $\RR (\C) = \{x_1, \ldots, x_k\}$,
    \begin{enumerate}[label=(\roman*), start=2]
        \item $x_i x_{i+1} | x_{i+2} x_{i+3} \in \R (\N) $ for all $i \in \{2, \ldots, k - 3 \}$;
        \item $(x_1 x_i x_{i+1} x_{i+2}) \in \R (\N)$ for all $i \in \{2, \ldots, k-2\}$.
    \end{enumerate}
\end{definition}

\CorrB{Note that the representative quartet profiles are uniquely determined and do not depend on the specific representation of the circular or 2-circular order (i.e., up to reversal and cyclic permutation, and reversal, respectively), as all such representations yield the same set.} As an example, consider the network in \cref{fig:repr_quartets}. Then, the quartet profiles $x_1 x_7 | x_9 x_{11}$, $x_1 x_2 | x_5 x_7$ and $(x_4 x_5 x_7 x_1)$ are representative quartet profiles of types (i), (ii) and (iii), respectively. Whereas, for example, $x_1 x_2 | x_5 x_{11}$, $x_1 x_5 |x_2 x_3$ and $(x_4 x_5 x_7 x_2)$ are not representative.

\begin{figure}[htb]
    \centering
    \includegraphics{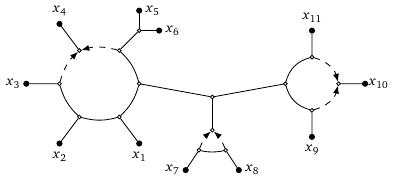}
    \caption{A semi-directed network on $\{x_1, \ldots, x_{11} \}$ used to illustrate the concept of a representative quartet profile.}
    \label{fig:repr_quartets}
\end{figure}

In general, there is exactly one representative quartet profile of type (i) for every nontrivial split of a network, and there are $(k-4) + (k-3) = 2k - 7$ representative quartet profiles of types (ii) and (iii) for every $k$-blob with $k\geq 4$. Hence, \CorrB{among all level-1 networks with a fixed number of leaves,} the size of $\R (\N)$ is \niels{largest} for a sunlet network---which has one large blob and no nontrivial splits. The following observation follows.
\begin{observation}\label{obs:representative_linear}
    If $\N$ is a semi-directed level-1 network on $n$ leaves, then $|\R (\N) | \leq 2n - 7$.
\end{observation}

\smallskip

In the remainder of this section we prove that the representative quartet profiles are sufficient to build up the complete set of quartet profiles, solely using dyadic inference rules. We formalize this by means of the following definition and stress that the triadic rule \ref{eq:R4*} is not being applied.

\begin{definition}[Dyadic closure]\label{def:dyadic_closure}
    Let $\Q$ be a set of quartet profiles. The set $\cl (\Q)$ of quartet profiles obtained from $\Q$ by exhaustively applying inference rules \ref{eq:R1}--\ref{eq:R4} is the \emph{dyadic closure} of $\Q$.
\end{definition}

\CorrB{Observe that the dyadic closure operation $\cl$ applied to a set of quartet profiles~$\Q$ indeed satisfies the standard closure axioms: 
\emph{extensivity} ($\Q \subseteq \cl(\Q)$), 
\emph{monotonicity} ($\Q_1 \subseteq \Q_2 \implies \cl(\Q_1) \subseteq \cl(\Q_2)$), and 
\emph{idempotency} ($\cl(\cl(\Q)) = \cl(\Q)$).}

% \CorrB{We first present a result that characterizes the structure of the subnetworks of a sunlet. We note that the result may be well known to readers familiar with level-1 networks, but we include it for completeness and ease of reference.}

% \begin{lemma}\label{lem:subnetworks_sunlet}
% Let $\N$ be a sunlet network on $\X$ with $x \in \X$ the leaf below the reticulation. Let $\Y \subseteq \X$ be such that $|\Y| \geq 3$. Then,
% \begin{enumerate}[label=(\roman*)]
% \item $\N|_\Y$ is a sunlet network if $x \in Y$;
% \item $\N|_\Y$ is an unrooted phylogenetic tree if $x \not\in Y$.
% \end{enumerate}
% \end{lemma}

\begin{lemma}\label{lem:dyadic_closure_sunlet}
    Let $\N$ be a sunlet network. If $\Q$ is a set of quartet profiles such that $\R (\N) \subseteq \Q \subseteq \Q (\N)$, then $\Q (\N) = \cl (\Q)$.
\end{lemma}
\begin{proof}
    Let $\X = \{x_1, \ldots, x_n\}$ be the leaves of~$\N$ such that $x_1$ is the leaf below the reticulation and such that $\N$ induces the circular ordering $(x_1, \ldots, x_n)$ of its leaves. See also \cref{fig:dyadic_closure_proof}a. 
    %Without loss of generality, we may assume that $\Q = \R (\N)$. 
    \CorrB{By monotonicity of the dyadic closure (i.e., $\Q_1 \subseteq \Q_2 \implies \cl(\Q_1) \subseteq \cl(\Q_2)$; see also the remark after \cref{def:dyadic_closure}), it suffices to consider the minimal case $\Q = \R(\N)$.}
    Throughout the proof, we denote by $\Q_t (\N)$ (resp. $\Q_c (\N)$) the trivial (resp. nontrivial) quartet profiles in $\Q( \N)$. We have already shown that the inference rules~\ref{eq:R1}-\ref{eq:R4} are valid for~$\N$ (\cref{prop:rule1_valid,lem:rule3_valid,lem:rule4_valid}), so $\cl(\Q) \subseteq \Q_t (\N) \cup \Q_c (\N)$. Thus, it suffices to show that $\cl(\Q) \supseteq \Q_t (\N) \cup \Q_c (\N)$. We prove the lemma by induction on $n$. The base cases with $n\leq 4$ follow trivially. 

    \CorrB{Let $u_i$ denote the non-leaf neighbor of $x_i$, $1 \leq i \leq n$. Note that $u_1$ is the reticulation. Let $T = \N|_{\{x_2, \ldots, x_n\}}$ and $\N' = \N|_{\{x_1, \ldots, x_{n-1}\}}$. All vertices except $x_1$ and $u_1$ lie on up-down paths between leaves of $T$, and all vertices except $x_n$ lie on such paths in $\N'$. Hence, by \cref{def:restriction}, $T$ is obtained by removing $x_1$ and $u_1$ and suppressing $u_2$ and $u_n$, while $\N'$ is obtained by removing $x_n$ and suppressing $u_n$. Thus, $T$ is a tree and $\N'$ is a sunlet.}
    
    %Since $T$ is a tree, $\R (T ) \subseteq \R(\N)$ consists of only trivial quartet profiles. 
    \CorrB{Since $T$ is a tree, each quartet profile in $\R(T)$ is a trivial quartet profile. Moreover, each quartet profile of $N$ not involving $x_1$ is also trivial. Combining this with \cref{def:repr_quartet}---which implies that the leaf set of a representative quartet profile in $T$ corresponds to the leaf set of a representative quartet profile in $N$---, we obtain $\R (T ) \subseteq \R(\N)$.}
    It follows from \cite[Lem.\,2]{erdos1999logs1} that exhaustively applying inference rules \ref{eq:R1} and \ref{eq:R2} to $\R (T )$ results in $\Q(T)$: the set of all quartet profiles induced by the tree~$T$. Since $\Q(T) = \Q_t (\N)$, we get that $\cl(\Q) \supseteq \Q_t (\N)$.
    
    \CorrB{Since $N'$ is a sunlet and it induces the same circular ordering as $N$ except for the removal of the leaf~$x_n$, each quartet profile of $N'$ is a quartet profile of $N$. As before, in light of \cref{def:repr_quartet}, the leaf set of a representative quartet profile in $N'$ corresponds to the leaf set of a representative quartet profile in $N$. So, } 
    $\R(\N') \subseteq \R(\N) = \Q$. In particular, $\R(\N) = \R(\N') \cup \{ x_{n-3} x_{n-2} | x_{n-1} x_n, (x_1 x_{n-2} x_{n-1} x_n)\}$. Thus, by the induction hypothesis, $\cl (\Q) \supseteq \Q (\N')$. Since we also already have that $\cl(\Q) \supseteq \Q_t (\N)$, it remains to show that $q = (x_1 x_i x_j x_n) \in \cl (\Q)$ for all $1<i<j<n$. We consider two cases. (1): $j < n-1$. Then, both $(x_1 x_i x_j x_{n-1})$ (since it is in $\Q (\N')$) and $x_i x_j | x_{n-1} x_n$ (since it is in $\Q_t (\N)$) are in $\cl (\Q)$. Applying rule~\ref{eq:R4} to these two quartet profiles shows that $q \in \cl (\Q)$. (2): $j = n-1$. Then, we may also assume that $i<n-2$, otherwise $q$ is in $\Q$ and thus trivially in $\cl (\N)$. Both $(x_1 x_{n} x_{n-1} x_{n-2})$ (since it is in $\R(\N)$) and $x_n x_{n-1} | x_{n-2} x_i$ (since it is in $\Q_t (\N)$) are in $\cl (\Q)$, so applying rule~\ref{eq:R4} shows that $q \in \cl (\Q)$.
\end{proof}

To prove the main theorem in this section (\cref{thm:dyadic_closure}), the following lemma will first show a slightly weaker result about the dyadic closure. In particular, given a semi-directed level-1 network $\N$, we call a quartet profile $q \in \Q(\N)$ a \emph{spanning quartet profile} if there exists a circular order~$\C$ or 2-circular order~\CorrB{$\D$} induced by an edge or a blob of $\N$
% \todo{MJ: what is a circular order induced by $\N$ (as opposed to an edge or blob)? You are right, thnx! Should be an edge or a blob of $N$.} 
such that $\LL (q) \subseteq \RR (\C)$ \CorrB{or $\LL (q) \subseteq \RR (\D)$, respectively}. That is, $q$ is a quartet profile containing only representative leaves of some induced (2-)circular order of $\N$. We denote the set of all spanning quartet profiles of~$\N$ by $\QS (\N)$ and note that $\R (\N) \subseteq \QS (\N)$.

\begin{figure}[htb]
    \centering
    \includegraphics{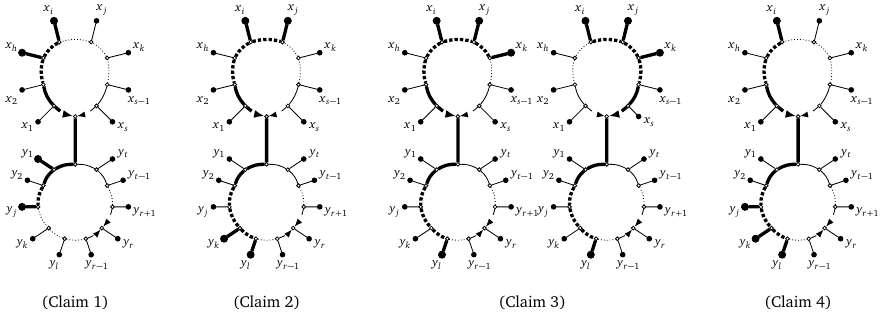}
    \caption{\CorrB{Illustration supporting the four claims in the proof of \cref{lem:spanning_quartet_profiles}. Each subfigure depicts one option of how the considered quartet profile is displayed in the shown network. The quartet profiles are (from left to right): \(x_h x_i | y_1 y_j\), \(x_i x_j | y_k y_l\), \((x_i x_j x_k y_l)\), and \(x_i y_j | y_k y_l\). The subfigure corresponding to Claim~3 shows both quartets \(x_i x_j | x_k y_l\) and \(x_j x_k | x_i y_l\) contributing to the nontrivial quartet profile \((x_i x_j x_k y_l)\).}}
    \label{fig:spanning_quartets_proof}
\end{figure}

\begin{lemma}\label{lem:spanning_quartet_profiles}
    Let $\N$ be a semi-directed level-1 network. If $\Q$ is a set of quartet profiles such that $\R (\N) \subseteq \Q \subseteq \Q (\N)$, then $\QS (\N) \subseteq \cl (\Q)$. 
\end{lemma}
\begin{proof}
    As in the previous lemma, \CorrB{it suffices to consider the minimal case $\Q = \R (\N)$ due to the monotonicity of the dyadic closure. We may assume that the number of leaves~$n$ of $\N$ is at least~4, otherwise the result follows trivially.} Since the presence of 2-cycles does not influence the quartet profiles, we will also assume that $N$ has no 2-cycles. We prove the theorem by induction on the number of nontrivial splits.

    \CorrB{The base case, where $\N$ has no nontrivial splits, implies that $\N$ has a single internal blob to which all leaves are attached. Further, this internal blob cannot be a single vertex, since $\N$ is binary and $n\geq 4$. Thus, $\N$ is a sunlet and the result then follows from the previous lemma.}

    Now assume that $\N$ has at least one nontrivial split.
    Let $B_Y$ be a \emph{lowest} internal blob of $\N$, i.e., $B_Y$ has only one incident nontrivial cut-edge~$uv$ (with $v$ in~$B_Y$) such that $v$ is not the reticulation of $B_Y$. Let $B_X$ be the blob that contains $u$. See \cref{fig:dyadic_closure_proof}b and c for an illustration. For simplicity, we also assume that both $B_X$ and $B_Y$ are nontrivial, i.e., if they are 3-blobs they are not single vertices. Note that this transformation does not change any of the quartet profiles. Let $\D = (X_1 , \ldots , X_s )| ( \{y_1\} , \ldots , \{y_t \} )$ be the induced 2-circular order of $uv$, where $Y = \{y_1, \ldots, y_t\}$, $X = X_1 \cup \ldots \cup X_s$ and $\RR( \D ) = \{x_1, \ldots, x_s, y_1, \ldots, y_t\}$ contains the representative leaves. Here, $y_r$ with $1 \leq r \leq t$ is the leaf below the reticulation of $B_Y$. Note that $\N$ can be of two different types (see again \cref{fig:dyadic_closure_proof}), depending on whether the leaves in $Y$ are below the reticulation of $B_X$ (type~1, shown in \cref{fig:dyadic_closure_proof}b), or whether some set $X_p$ (with $1 \leq p \leq s$) is below the reticulation of $B_X$ (type~2, shown in \cref{fig:dyadic_closure_proof}c). 
    
    %We consider Case~(a) here, with Case~(b) being similar and deferred to the appendix.

    Let $\N_X = \N|_{X \cup \{y_1\}}$ and let $\N_Y = \N|_{Y \cup \{x_{h}\}}$, where $x_h \in \{x_1, \ldots, x_s \}$\CorrB{, $1 \leq h \leq s$,} is the leaf topologically closest to the blob~$B_Y$, \CorrB{with ties resolved by choosing the leaf with the smallest index}. Then, $\R (\N) = \R (\N_X) \cup \R (\N_Y) \cup \{x_1 x_s | y_1 y_t \}$. \CorrB{Also observe that $\N_Y$ is a sunlet inducing a unique circular order~$\C$ of all singleton sets of leaves in $Y \cup \{x_h\}$. Then, each leaf of $N_Y$ is a representative leaf of the singleton set it is contained in. So, each quartet profile of $N_Y$ is a spanning quartet profile, i.e., $\QS(N_Y) = \Q (N_Y)$. Combining this with} the induction hypothesis, we then get that $\QS(\N_X) \cup \Q (\N_Y) \cup \{x_1 x_s | y_1 y_t \} \subseteq \cl (\Q) $. We now prove four claims, showing that certain quartet profiles in $\QS(\N)$ are in $\cl(\Q)$. Together with the fact that \CorrB{$\QS (\N_X)\subseteq \cl (\Q)$ and $\Q(\N_Y) \subseteq \cl (\Q)$}, Claims~2-4 then show that $\QS(\N) \subseteq \cl (\Q)$. \CorrB{For each claim, \cref{fig:spanning_quartets_proof} illustrates how the corresponding quartet profile is displayed in the network, in the case where the network is of type~1 and each set \(X_i\) is a singleton \(\{x_i\}\). For the figure, we assume \(h < i\) and \(r - 1 > l\), such that \(x_h\) appears to the left of the leaves \(x_i, x_j, x_k\), and \(y_r\) appears to the right of the leaves \(y_j, y_k, y_l\) in the circular ordering induced by the shown planar representation.}

    \textbf{Claim 1:} \emph{$q \in \cl (\Q)$ if $\LL (q) = \{x_h, x_i, y_1, y_j \}$ with $1 \leq i \leq s$ such that $i \neq h$ and $1 < j \leq t$.} Without loss of generality, we assume that $1 \leq h < i \leq s$. Note that every such $q$ has the form $x_h x_i | y_1 y_j$ \CorrB{(see \cref{fig:spanning_quartets_proof})}. Observe that $q_1 = x_1 x_s | y_1 y_t$ is in $\cl (\Q)$. If $p \in \{1, i, s\}$ then $q_2= (x_1 x_i x_s y_1)$ is in $\cl (\Q)$ (Case~1(i)). Otherwise, without loss of generality, we may assume that $q_2 = x_1 x_i | x_s y_1$ is in $\cl (\Q)$ (Case~1(ii)). In Case~1(i) (resp. Case~1(ii)), we can apply rule~\ref{eq:R3} (resp. rule~\ref{eq:R1}) to $q_1$ and $q_2$, to obtain that $q_3 = x_1 x_i | y_1 y_t$ is in $\cl (\Q)$. (Note that this follows trivially if $i=s$, since then $q_3 = q_1$). If $h = 1$, the proof is done, so we assume that $h > 1$. An analogous proof then shows that $q_4 = x_1 x_h | y_1 y_t$ is in $\cl (\Q)$. Applying rule~\ref{eq:R2} to $q_3$ and $q_4$ then shows that $q$ is in $\cl (\Q)$.

    \textbf{Claim 2:} \emph{$q \in \cl (\Q)$ if $\LL (q) = \{x_i, x_j, y_k, y_l \}$ with $1 \leq i < j \leq s$ and $1 \leq k < l \leq t$.} Again, every such $q$ has the form $x_i x_j | y_k y_l$ \CorrB{(see \cref{fig:spanning_quartets_proof})}. First assume that $k > 1$, the proof being analogous for $k = 1$ by changing $y_1$ to $y_2$ throughout. By the previous claim, we know that $q_1 = x_h x_i | y_1 y_k$ and $q_2 = x_h x_j | y_1 y_k$ are in $\cl (\Q)$. Applying rule~\ref{eq:R2} shows $q_3 = x_i x_j | y_1 y_k \in \cl (\Q)$. (This follows directly if $h \in \{i, j\}$.) Similarly, it follows that $q_4 = x_i x_j | y_1 y_l$ is in $\cl (\Q)$. Then, again applying rule~\ref{eq:R2} to $q_3$ and $q_4$ gives that $q = x_i x_j | y_k y_l$ is in $\cl (\Q)$.

    \textbf{Claim 3:} \emph{$q \in \cl (\Q)$ if $\LL (q) = \{x_i, x_j, x_k, y_l \}$ with $1 \leq i < j < k \leq s$ and $1 \leq l \leq t$.} If $\N$ is of type~1 or of type~2 with $p \in \{i,j,k\}$, then $q$ has the form $(x_i x_j x_k y_l)$ (Case~3(i), \CorrB{see \cref{fig:spanning_quartets_proof}}). If $\N$ is of type~2 with $p \not\in \{i,j,k\}$, then we may assume without loss of generality that $q$ has the form $x_i x_j | x_k y_l$ (Case~3(ii)). Next, observe that in Case~3(i) (resp. Case~3(ii)) $q_1=(x_i x_j x_k y_1)$ (resp. $q_1 = x_i x_j | x_k y_1$) is in $\Q (\N_X) \subseteq \cl (\Q)$ and that $q_2 = x_j x_k | y_1 y_l$ is in $\cl (\Q)$ by Claim~2. Then, applying rule~\ref{eq:R4} (resp. rule~\ref{eq:R2}) to $q_1$ and $q_2$ proves that $q \in \cl (\Q)$.

    \textbf{Claim 4:} \emph{$q \in \cl (\Q)$ if $\LL (q) = \{x_i, y_j, y_k, y_l \}$ with $1 \leq i  \leq s$ and $1 \leq j < k < l \leq t$.} If $r \in \{j, k, l\}$, then $q$ has the form $(x_i y_j y_k y_l)$ (Case~4(i)). Otherwise, we may assume without loss of generality that $q$ has the form $x_i y_j | y_k y_l$ (Case~4(ii), \CorrB{see \cref{fig:spanning_quartets_proof}}). In Case~4(i) (resp. Case~4(ii)), $q_1 = (x_h y_j y_k y_l)$ (resp. $q_1 = x_h y_j | y_k y_l$) is in $\Q(\N_Y) \subseteq \cl (\Q)$, so we may assume that $i \neq h$. Then, $q_2 = x_h x_i | y_k y_l$ is in $\cl (\Q)$ by Claim~2. Applying rule~\ref{eq:R4} (resp. rule~\ref{eq:R4}) to $q_1$ and $q_2$ proves that $q \in \cl (\Q)$.
\end{proof}

\begin{theorem}\label{thm:dyadic_closure}
    Let $\N$ be a semi-directed level-1 network. If $\Q$ is a set of quartet profiles such that $\R (\N) \subseteq \Q \subseteq \Q (\N)$, then $ \Q (\N) = \cl (\Q)$.
\end{theorem}
\begin{proof}
    \CorrB{Since the inference rules~\ref{eq:R1}-\ref{eq:R4} are valid for~$\N$ (\cref{prop:rule1_valid,lem:rule3_valid,lem:rule4_valid})we have that $\Q(\N) \supseteq \cl (\Q)$. It remains to show that $\Q (\N) \subseteq \cl (\Q)$.}
    
    \CorrB{By \cref{lem:spanning_quartet_profiles}, we have $\QS(\N) \subseteq \cl(\Q)$. Using monotonicity and idempotency of the dyadic closure (see the paragraph after \cref{def:dyadic_closure}), this implies $\cl(\QS(\N)) \subseteq \cl(\cl(\Q)) = \cl(\Q)$. 
    Therefore, if we show $\Q(\N) \subseteq \cl(\QS(\N))$, the inclusion $\Q(\N) \subseteq \cl(\Q)$ will follow. So, it suffices to prove $\Q (\N) \subseteq \cl (\Q)$ for $\Q = \QS(\N)$.}
    We prove the theorem by induction on the number of nontrivial splits, with the base case again following from \cref{lem:dyadic_closure_sunlet}. 
    
    Consider the exact same setup as in the proof of \cref{lem:spanning_quartet_profiles}. Then, by the induction hypothesis, we have that $\Q (\N_X) \cup \Q(\N_Y) \subseteq \cl (\Q)$. We slightly extend the notation by using $x'_i$ to indicate an arbitrary leaf in the set $X_i$, as opposed to the notation $x_i$ that was used for the representative leaf of $X_i$. With slight abuse of notation, when considering leaves $x'_i$ and $x'_j$, we allow $i=j$, in which case the leaves are meant to be two distinct leaves in $X_i = X_j$. We now prove three claims for quartet profiles $q \in \Q(\N)$.

    \textbf{Claim 5:} \emph{$q \in \cl (\Q)$ if $\LL (q) = \{x'_i, x'_j, y_k, y_l \}$ with $1 \leq i, j  \leq s$ and $1 \leq k < l \leq t$.} \CorrB{Observe that the subnetwork $\N|_{\LL (q)}$ always induces the split $x'_i x'_j | y_k y_l$ (see also \cref{fig:dyadic_closure_proof}b and c). Hence, every displayed quartet of $\N$ induces the split $x'_i x'_j | y_k y_l$, implying that the quartet profile~$q$} has the form $x'_i x'_j | y_k y_l$. Since $x_i x_j | y_k y_l$ is in $\cl (\Q)$, we can assume without loss of generality that $x'_i \neq x_i$. We consider two cases.
    
    \emph{Case~(i):} $i \neq j$. Observe that $x_i x'_i | x_j y_1$ is in $\Q (\N_X) \subseteq \cl (\Q)$ and $x_i x_j | y_1 y_l$ is in $\cl (\Q)$. Then, by rule~\ref{eq:R1}, $q_1 = x_i x'_i | y_1 y_l$ and $q_2 = x'_i x_j | y_1 y_l$ are in $\cl (\Q)$. By symmetry, $q_3 = x'_j x_i | y_1 y_l$ is then also in $\cl (\Q)$. (Note that this follows trivially if $x_j = x'_j$.) Applying rule~\ref{eq:R2} to $q_1$ and $q_3$ gives that $q_4 = x'_i x'_j | y_1 y_l$ is in $\cl (\Q)$. If $k=1$, we are done. Otherwise, again using symmetry, we get that $q_5 = x'_i x'_j | y_1 y_k$ is in $\cl (\Q)$. Applying rule~\ref{eq:R2} to $q_4$ and $q_5$ gives that $q$ is in $\cl (\Q)$. 
    
    \emph{Case~(ii):} $i = j$. Then, $x'_i$ and $x'_j$ are two distinct leaves in $X_i = X_j$. Let $m \neq i, j$ be arbitrary, and note that $x_m$ always exists since $i=j$. Observe that $x_i x'_i | x_m y_1$ is in $\Q (\N_X) \subseteq \cl (\Q)$ and $x_i x_m | y_1 y_l$ is in $\cl (\Q)$. Then, by rule~\ref{eq:R1}, $q_1 = x_i x'_i | y_1 y_l$ and $q_2 = x'_i x_m | y_1 y_l$ are in $\cl (\Q)$. Similarly, by swapping $x'_i$ for $x'_j$, $q_3 = x'_j x_m | y_1 y_l$ is in $\cl (\Q)$ (which is trivial if $x'_j = x_i$). Applying rule~\ref{eq:R2} to $q_2$ and $q_3$ gives that $q_4 = x'_i x'_j | y_1 y_l$ is in $\cl (\Q)$. The rest of the proof is analogous to Case~(i).

    \textbf{Claim 6:} \emph{$q \in \cl (\Q)$ if $\LL (q) = \{x'_i, x'_j, x'_k, y_l \}$ with $1 \leq i, j, k  \leq s$ and $1 \leq l \leq t$.} The proof of this claim is analogous to the proof of Claim~3 in the proof of \cref{lem:spanning_quartet_profiles}, but instead relying on Claim~5 instead of Claim~2.

    \textbf{Claim 7:} \emph{$q \in \cl (\Q)$ if $\LL (q) = \{x'_i, y_j, y_k, y_l \}$ with $1 \leq i \leq s$ and $1 \leq j < k < l \leq t$.} Again, the proof of this claim is analogous to the proof of Claim~4 in the proof of \cref{lem:spanning_quartet_profiles}, but instead relying on Claim~5 instead of Claim~2.

    To conclude the proof of the theorem, observe that any quartet profile in $\Q (\N) \setminus (\Q (\N_X) \cup \Q(\N_Y) \cup \QS(\N))$ is shown to be in $\cl (\Q)$ by one of the three claims.
\end{proof}

\CorrB{Recall that \cref{prop:qp_reconstruct} stated that given a semi-directed level-1 network $\N$, the mixed graph~$\NN$ obtained from $\N$ by suppressing 2-cycles, 3-cycles and reticulation vertices in 4-cycles, is uniquely characterized by quartet profiles. Combining the previous theorem with \cref{prop:qp_reconstruct} and \cref{obs:representative_linear}, we obtain the following corollary, where we use the term \emph{quarnet} for a 4-leaf subnetwork of a network.}\footnote{In \cite{frohn2024reconstructing} it was shown that any algorithm that sequentially adds leaves to reconstruct a semi-directed level-1 network from quarnets must rely on $\Omega (n \log n)$ quarnets in the worst case. Although part~(b) of \cref{cor:linear_encoding} may appear to contradict this bound, it does not violate the result presented there. In particular, the bound from \cite{frohn2024reconstructing} is a lower bound on the \emph{query complexity}, meaning that for any algorithm relying on queries to an oracle (which returns quarnets of the network to be reconstructed), there exists an adversarial strategy that forces the algorithm to make $\Omega (n \log n)$ queries. In contrast, \cref{cor:linear_encoding}b instead says that once a network is known, a linear sized set of quarnets exists that completely characterizes it.}

\begin{corollary}\label{cor:linear_encoding}
Let $\N$ be a semi-directed level-1 network on $n \geq 4$ leaves. Then,
\begin{enumerate}[label={(\alph*)},noitemsep]
    \item there exists a set of $\bigO(n)$ quartet profiles of~$\N$ that uniquely characterizes $\NN$;
    \item there exists a set of $\bigO(n)$ quarnets of~$\N$ that uniquely characterizes~$\N$.
\end{enumerate}
\end{corollary}
\begin{proof}
    \CorrB{By \cref{obs:representative_linear}, there are $\bigO(n)$ representative quartet profiles. By the previous theorem, the full set of quartet profiles can be obtained from these, and by \cref{prop:qp_reconstruct}, this set uniquely characterizes~$\NN$. This yields part~(a).} 

    \CorrB{For part~(b), we first note that by \cref{prop:displayed_subnet_commutative}, a quartet profile on leaf set $\{a,b,c,d\}$---that is, the set of 4-leaf subtrees induced by $\{a,b,c,d\}$ in all displayed trees of $\N$---can equivalently be obtained as the set of displayed trees of the 4-leaf subnetwork / quarnet $\N|_{\{a,b,c,d\}}$. Hence, by part~(a), there exists a set of $\bigO( n)$ quarnets that uniquely charachterizes~$\NN$.}

    \CorrB{To recover the full network structure~$\N$ (i.e., to determine the resolution of 2-cycles, 3-cycles, and reticulation in 4-cycles), it suffices to consider specific quarnets. For each 4-cycle, a single quarnet with one leaf on each of the four sides of the 4-cycle determines the placement of the reticulation. For each 3-blob in~$\NN$ inducing a circular order $(X_1,X_2,X_3)$, one quarnet with one leaf from each of two parts and two leaves from the remaining part suffices to determine its structure.}

    \CorrB{Finally, 2-cycles can only be placed on cut-edges of~$\N$, as otherwise the resulting network would not be level-1. To detect 2-cycles on a trivial cut-edge, it suffices to consider any quarnet containing the corresponding leaf. For a nontrivial cut-edge inducing a 2-circular order $(X_1,\ldots,X_k)\,|\,(Y_1,\ldots,Y_\ell)$, it suffices to consider a quarnet containing two leaves from (distinct parts of) the $X$-side and two leaves from (distinct parts of) the $Y$-side. The corresponding quarnet then contains a nontrivial cut-edge, and the presence or absence of a 2-cycle on this cut-edge in the quarnet directly reflects its presence or absence in the original network. Consequently, this determines whether a 2-cycle is present on the given cut-edge of~$\N$.}
    
    \CorrB{Altogether, since there are $\bigO (n)$ blobs and cut-edges in~$\N$, this yields a linear-sized set of quarnets that uniquely characterizes~$\N$.}
\end{proof}

\section{The Dyadic Closure Method}\label{sec:DCM}

In this section we present an algorithm that reconstructs (most of) a semi-directed level-1 network from a sequence alignment on $\mathcal{X}$. We follow a strategy similar as in \cite{erdos1999logs1} and as in the analysis of the `short quartet support method' in~\cite{warnow2001absolute}. 

As a first ingredient we present the \dcnc (\textsc{Dyadic Closure Network Construction}) method (see \cref{alg:dcnc}), which takes as input a set of quartet profiles, and outputs a level-1 network, \textsc{Inconsistent} (if two quartet profiles are incompatible), or \textsc{Insufficient} (if the quartet profiles do not provide enough information to construct a network). From \cref{thm:dyadic_closure} we can draw an immediate conclusion for the DCNC method, formulated in \cref{cor:DCNC1}. \CorrB{Recall from \cref{prop:qp_reconstruct} that we use the notation~$\NN$ to denote the mixed graph obtained from a semi-directed level-1 network~$\N$ by suppressing 2-cycles, 3-cycles and reticulation vertices in 4-cycles.}

\begin{algorithm}[!t]
\caption{\textsc{Dyadic Closure Network Construction} (\dcnc) method}\label{alg:dcnc}
\Input{a set of quartet profiles $\Q$ on $\X$}%\mf{quartet profiles (Definition 2.4) need a network to be well-defined NH: You are right, i've updated the def to make it work}}
\Output{a semi-directed level-1 network $\N'$ on $\X$ with 2-cycles, 3-cycles and reticulation vertices in 4-cycles suppressed, \textsc{Insufficient}, or \textsc{Inconsistent}}
compute the dyadic closure $\cl (\Q)$ analogous to \cite{erdos1999logs1} \\
\If{$\cl (\Q)$ contains more than one quartet profile for some set of four leaves}{\Return{\textsc{Inconsistent}}}
\ElseIf{$|\cl (\Q)| < {n \choose 4}$}{\Return{\textsc{Insufficient}}}
\Else{
$\N' \gets$ semi-directed level-1 network with 2-cycles, 3-cycles and reticulation vertices in 4-cycles suppressed, obtained by applying the algorithm from \cite{frohn2024reconstructing} to $\Q$ \\
\Return{$\N'$}
}
\end{algorithm}

%\begin{proposition}\label{prop:DCNC}
    %Let $\Q$ be a set of quartet profiles on $\X=\{x_1,\dots,x_n\}$ and let $\N$ be a semi-directed level-1 network on~$\X$. Then, algorithm \dcnc applied to $\Q$ can be implemented to run in $\bigO (n^5)$ time and if $\R (\N) \subseteq \Q \subseteq \Q (\N)$ it returns the skeleton $\NN$.
%\end{proposition}
\begin{corollary}\label{cor:DCNC1}
     Let $N$ be a semi-directed level-1 network on $\mathcal{X}$ and let $\Q$ be a set of quartet profiles such that $\Q_R(N)\subseteq\Q\subseteq\Q(N)$. Then, the \CorrB{network~$\N'$ returned by the DCNC method satisfies $\N' = \NN$}.
\end{corollary}

\begin{algorithm}[!t]
\caption{\textsc{Level-1 Quartet Profile Construction} (\lqpc) method}\label{alg:L1QPC}
\Input{a real number $w\geq 0$, an empirical distance matrix $d$ from a sequence alignment on $\mathcal{X}$}
\Output{a set of quartet profiles $\mathcal{Q}_w$ of width at most $w$}
\For{$\{x_i,x_j,x_k,x_l\}\subseteq\mathcal{X}$ with width at most $w$ and $d_{ij}+d_{kl}\leq d_{il}+d_{jk}\leq d_{ik}+d_{jl}$}{
$\Delta_1 \leftarrow (d_{il} + d_{jk}) - (d_{ij} + d_{kl})$\;
$\Delta_2 \leftarrow (d_{ik} + d_{j\ell})- (d_{il} + d_{jk}) $\;
%$\alpha_l\leftarrow(d_{ij}+d_{kl})/(d_{il}+d_{jk})$\;
%$\alpha_k\leftarrow(d_{il}+d_{jk})/(d_{ik}+d_{jl})$\;
%\If{$\alpha\leq\alpha_k\cdot\alpha_l$}
\If{$\Delta_1=\Delta_2$}{Add trivial quartet profiles $x_ix_j|x_kx_l$, $x_ix_k|x_jx_l$ and $x_ix_l|x_jx_k$ to $\mathcal{Q}_w$.}
%\If{$\alpha_l<\alpha\leq\alpha_k$}
\ElseIf{$\Delta_1>\Delta_2$}{Add trivial quartet profile $x_ix_j|x_kx_l$ to $\mathcal{Q}_w$.}
%\If{$\alpha_k<\alpha$}
\Else{Add quartet profile $(x_ix_jx_kx_l)$ to $\mathcal{Q}_w$.}
} 
\Return{$Q_w$}
\end{algorithm}

\begin{algorithm}[!b]
\caption{\textsc{Dyadic Closure} (\dc) method}\label{alg:dcm}
\Input{a sequence alignment $S$ on $\X$}
\Output{a semi-directed level-1 network $\N'$ on $\X$ with 2-cycles, 3-cycles and reticulation vertices in 4-cycles suppressed, or \textsc{Fail}}
\For{$x_i,x_j\in\mathcal{X}$}{ 
    Compute dissimilarity score $h^{ij}$ from $\mathcal{S}$\;
    $d_{ij}\leftarrow -\frac{1}{2}\log(1-2h^{ij})$\;
}
$\N' \gets$ \textsc{None} \\
\For(\tcp*[f]{perform a binary search}){dissimilarity score $w$ derived from $\mathcal{S}$}{
$\Q_w \gets \lqpc (w,d)$\\
\If{$\dcnc (\Q_w) = \textsc{Insufficient}$}{increase $w$}
\ElseIf{$\dcnc (\Q_w) = \textsc{Inconsistent}$}{decrease $w$}
\ElseIf{$\dcnc (\Q_w) = \NN_w$ for some semi-directed level-1 network $\N_w$}{
$\N' \gets \NN_w$\\
increase $w$
}
}
\If{$\N' = $ \textsc{None}}{
\Return{\textsc{Fail}}
}
\Return{$\N'$}
\end{algorithm}

\begin{definition}[width of a quartet]
    For a sequence alignment $\mathcal{S}$ on $\X$ the \emph{width} of any \niels{quartet tree or quartet profile on four leaves} $Q\subseteq\X$ is defined as the maximum dissimilarity score $h^{ij}$ of sequences $s^i$ and $s^j$ among all $x_i,x_j\in Q$. %The width of a quartet profile is the width of the underlying quartet.
\end{definition}

% Our second ingredient, the \lqpc (\textsc{Level-1 Quartet Profile Construction}) method (see \cref{alg:L1QPC}), connects a sequence alignment on $\mathcal{X}$ to a set of quartet profiles of bounded width\todo{MJ: at most a certain width ?}. Intuitively, increasing the input parameter $w$ for the \lqpc method decreases the confidence that all inferred quartet profiles are correct. On the other hand, a small value of $w$ increases this confidence, but forces \lqpc to return only a small set of quartet profiles which might be insufficient to guarantee the inference of the dyadic closure (see \cref{thm:dyadic_closure}).\todo{For MJ: rewrite these sentences?}
Our second ingredient, the \lqpc (\textsc{Level-1 Quartet Profile Construction}) method (see \cref{alg:L1QPC}), %connects a sequence alignment on $\mathcal{X}$ to a set of quartet profiles
\niels{infers a set of quartet profiles from a sequence alignment on $\mathcal{X}$}, where the quartets considered have bounded width.
Intuitively, increasing the input parameter $w$ for the \lqpc method increases the number of quartet profiles that are returned. This decreases the confidence that all inferred quartet profiles are correct, as there is a higher probability that at least one profile is incorrect. On the other hand, a small value of $w$ increases this confidence, but forces \lqpc to return only a small set of short-width quartet profiles which might be insufficient to guarantee the inference of the dyadic closure (see \cref{thm:dyadic_closure}).

Now, we are ready to state the main algorithm (see \cref{alg:dcm}). The \dc method does a binary search over the space of parameters $w$ to find a large enough set of quartet profiles such that the representative quartet profiles are also returned, but small enough such that the dyadic closure can be inferred correctly. The algorithm iteratively applies the previous algorithms \lqpc and \dcnc, making it a \niels{polynomial-time} %efficient 
algorithm:
\begin{proposition}\label{prop:DCNC2}
    The \dc method can be implemented to run in $\bigO (n^5 \cdot \log n)$ time.
\end{proposition}
\begin{proof}
    Since the \dc method performs a binary search over $n^2$ distances, both \lqpc and \dcnc are called at most $\bigO (\log n )$ times. Clearly, \lqpc takes $\bigO (n^4)$ time, \niels{whereas} \dcnc can be implemented to run in $\bigO (n^5)$ time. Specifically, the implementation outlined in \cite[Thm.\,5]{erdos1999logs1} shows that the dyadic closure (using only rules \ref{eq:R1} and \ref{eq:R2}) can be computed in $\bigO (n^5)$ time. Since our additional rules \ref{eq:R3} and \ref{eq:R4} are also dyadic, this time complexity remains unchanged. Furthermore, the algorithm in \cite{frohn2024reconstructing} runs in $\bigO (n^2)$ time (\cref{prop:qp_reconstruct}), proving that $\dcnc$ indeed takes $\bigO (n^5)$ time.
\end{proof}

In the next section we show that the \dc method works as intended. This means, we will show that there exists a set of quartet profiles $\mathcal{Q}_w$ and a semi-directed level-1 network $\NN$ (with 2-cycles, 3-cycles and reticulation vertices in 4-cycles suppressed) returned by Algorithm~\ref{alg:dcm} such that
\begin{enumerate}
    \item $\mathcal{Q}_R(\N)\subseteq\mathcal{Q}_w$,
    \item $\mathcal{Q}_w\subseteq\mathcal{Q}(\N)$,
    \item the \lqpc method finds $Q_w$ during the binary search with high probability.
\end{enumerate}
The following technical result establishes condition~(1). We will extend this construction to show that conditions~(2) and~(3) hold with high probability given a sufficient number of sites $k$ in the sequence alignment~$\mathcal{S}$ on $\mathcal{X}$.

\begin{lemma}\label{lem::Qw}
    Let $\N=(V,E)$ be a semi-directed level-1 network on $\mathcal{X}$, $\mathcal{S}$ be a sequence alignment on $\mathcal{X}$ and let
    \begin{align*}
        w>\frac{1}{2}\left(1-\min_{e\in E}\theta(e)^{\text{depth}(\N)+lc(\N)}\right).
    \end{align*}
    Assume for all $x_i, x_j \in \X$ that if $\mathbb{P}\left[Y(x_i)\neq Y(x_j)\right]< w$ in each displayed tree of $\N$, then $h^{ij}<w$, where $h^{ij}$ is the dissimilarity score between the sequences in $\mathcal{S}$ associated to $x_i$ and $x_j$. Then there exists $\epsilon>0$ such that
    \begin{align*}
        \mathcal{Q}_R(\N)\subseteq\mathcal{Q}_{w-\epsilon}
    \end{align*}
    for the set of quartet profiles $\mathcal{Q}_{w-\epsilon }$ of width at most $w-\epsilon$.
    % {2. Does "any displayed tree" means "at least one displayed tree" or "each displayed tree"? 3. Should width be defined for quartets or quartet profiles, or both?}
\end{lemma}
\begin{proof}
We consider quartet profiles of types (i), (ii) and (iii). First, let $x_1x_k|y_1y_l\in\mathcal{Q}_R(N)$ and let $\T$ be any displayed tree of $N$. Then, the path $e_1e_2\dots e_m$ connecting taxa $z_i,z_j\in\{x_1,x_k,y_1,y_l\}$ in $\T$ yields
\begin{align*}
    \mathbb{P}\left[Y(z_i)\neq Y(z_j)\right]&=\frac{1}{2}\left(1-\prod_{i=1}^m(1-2p(e_i))\right)\leq\frac{1}{2}\left(1-\min_{e\in E}\theta(e)^{\text{depth}(\T)+lc(\N)}\right)\\
    &\leq\frac{1}{2}\left(1-\min_{e\in E}\theta(e)^{\text{depth}(\N)+lc(\N)}\right)<w.
\end{align*}
We assumed that $h^{ij}<w$ for $\mathbb{P}\left[Y(x_i)\neq Y(x_j)\right]<w$. Hence, there exists $\epsilon>0$ such that
\begin{align*}
    x_1x_k|y_1y_l\in\mathcal{Q}_{w-\epsilon}.
\end{align*}
Next, let $x_ix_{i+1}|x_{i+2}x_{i+3}\in\mathcal{Q}_R(N)$. Then, similar to our first case
\begin{align*}
    \mathbb{P}\left[Y(z_i)\neq Y(z_j)\right]\leq\frac{1}{2}\left(1-\min_{e\in E}\theta(e)^{\text{depth}(\N)}\right)<w,
\end{align*}
leading to the same conclusion. Finally, $(x_1x_ix_{i+1}x_{i+2})\in\mathcal{Q}_R(N)$ is analogous to our first case.
\end{proof}

\section{Performance of the DC method under the CFN model}\label{sec:DCM_analysis}
First, we provide sufficient conditions for the \lqpc method to return the correct set of quartet profiles $\mathcal{Q}_w$. To this end, for a semi-directed level-1 network $\N$ under the CFN model we define a function $s$ which maps the reticulation edges $e$ of $\N$ to the positive real numbers by equations
\begin{align}
    \frac{\gamma(e)}{1-\gamma(e)}\cdot\frac{\log\min_{e'\in E}\theta(e')}{\log\max_{e'\in E}\theta(e')}+\epsilon&=
    \begin{cases}
        \frac{2s(e)^2-7s(e)}{2s(e)^2+s(e)-1} &\text{if } \frac{\gamma(e)}{1-\gamma(e)}\cdot\frac{\log\min_{e'\in E}\theta(e')}{\log\max_{e'\in E}\theta(e')}<1\\
        \frac{4s(e)^3+4s(e)^2-s(e)-1}{12s(e)^2+4s(e)+1} &\text{otherwise}
    \end{cases}\label{eq::def:s}
\end{align}
for a sufficiently small $\epsilon>0$ such that $s(e)>1/2$. Observe that function $s$ is well-defined because the left hand side of \cref{eq::def:s} is a positive real number defined by our model parameters and the right hand side either converges to $1$ or diverges for increasing $s(e)$.
%\begin{align*}
%    |\gamma\ell_2-(1-\gamma)\ell_1|&\leq\min\{\gamma\ell_2,(1-\gamma)\ell_1\}\\
%    \gamma\ell_2-(1-\gamma)\ell_1&\leq(1-\gamma)\ell_1\\
%    (1-\gamma)\ell_1-\gamma\ell_2&\leq\gamma\ell_2\\
%    \gamma\ell_2&\leq 2(1-\gamma)\ell_1\\
%    (1-\gamma)\ell_1&\leq 2\gamma\ell_2
%\end{align*}
\begin{restatable}{lemma}{lemLFPMsuccess}
\label{prop::L1FPMsuccess}
Let $\N$ be a semi-directed level-1 network on leafset $\mathcal{Y}=\{x_i,x_j,x_k,x_l\}$ with \niels{$(x_i x_j x_k x_l)\in\mathcal{Q}_R(\N)$,} %$(ijkl)\in\mathcal{Q}_R(\N)$, 
containing a cycle of length four and displayed trees $\T_1$ and $\T_2$. Let $e$ be the reticulation edge from $\N$ present in $\T_1$ and let $d$ be the matrix of dissimilarity scores on $\mathcal{Y}$ defined by a sequence alignment on $\mathcal{Y}$. 
%Let
%$$L(\gamma,\{\theta(e)\}_{e\in E})=\frac{\gamma}{1-\gamma}\cdot\frac{\log\min_{e\in E}\theta(e)}{\log\max_{e\in E}\theta(e)}$$ and 
Then, for the width $w$ of $\mathcal{Y}$, the L1QPC method infers \niels{$(x_i x_j x_k x_l)$} %$(ijkl)$ 
from $d$ if the inheritance probability~$\gamma(e)$ satisfies
%\begin{align*}
%    \frac{1-\gamma}{\gamma}\in\left[\frac{1}{2}\cdot\frac{\log\min_{e\in E}\theta(e)}{\log\max_{e\in E}\theta(e)},2\cdot\frac{\log\max_{e\in E}\theta(e)}{\log\min_{e\in E}\theta(e)}\right]
%\end{align*}
 $D=\gamma(e) D(\T_1)+(1-\gamma(e))D(\T_2)$ such that
%there exists 
%\begin{align*}
%    \alpha_k\in\left(0,1-\max\left\{\frac{|w_4-w_2|}{w_1},\frac{|w_3-w_1|}{w_2}\right\}\right]
%\end{align*} 
%such that
\begin{align}
    %\left|d_r-D_r\right|&<\frac{1}{2}\min\limits_{a,b,c\in\{1,\dots,4\},\,a\neq b\neq c}\{w_a+w_b-w_c\}~&~&\forall\,r\in\{ik,jl\},\,D_r\in\{D_r^{\min},D_r^{\max}\},\label{L1QPCsuccess::con1}\\
    %\max\left\{|d_r-D_r|\,:\,r\in\{ij,kl,il,jk,ik,jl\},\,D_r\in\{D_r^{\min},D_r^{\max}\}\right\}&<\frac{f}{2}.
    \max\left\{|d_r-D_r|\,:\,r\in\{ij,kl,ik,jl\}\right\}
    &<-\frac{\gamma(e)}{4s(e)}\log\max_{e'\in E}\theta(e'),\label{L1QPCsuccess::con1}\\
    \max\left\{|d_r-D_r|\,:\,r\in\{il,jk,ik,jl\}\right\}
    &<\frac{\gamma(e) -1}{4s(e)}\log\max_{e'\in E}\theta(e').\label{L1QPCsuccess::con2}
\end{align}
%\begin{align*}
 %   \max\{|d_r-D_r|\,:\,r\in\{ij,kl,il,jk,ik,jl\}\}<\frac{w_1w_2}{2(w_1+w_2)}
%\end{align*}
%\begin{align*}
%    |d_{ij}-\gamma D_{ij}^{\min}-(1-\gamma)D_{ij}^{\max}|<\frac{\gamma}{4}\left(D^{\min}_{ik}+D^{\min}_{jl}-D^{\min}_{ij}-D^{\min}_{kl}\right)
%\end{align*}
\end{restatable}

\CorrB{The proof of \cref{prop::L1FPMsuccess} is deferred to Appendix~\ref{sec:appendix} as it is rather technical.} Throughout this section we will make repeated use of the following Azuma-Hoeffding concentration inequalities~\cite{alon1992probabilistic}:
\begin{lemma}\label{AzumaHoeffding}
    Let $k$ be a positive integer, let $S$ be any set and let $Y=(Y_1,\dots,Y_k)$ be independent random variables taking values in $S$. For some $t\in\mathbb{R}$, let $L:S^k\to\mathbb{R}$ be a function such that $|L(a)-L(b)|\leq t$ whenever $a$ and $b$ differ in only one coordinate. Then, for $\lambda >0$,
    \begin{align*}
        \mathbb{P}\left[L(Y)-\mathbb{E}[L(Y)]\geq\lambda\right]\leq\exp\left(-\frac{\lambda^2}{2t^2k}\right)~~~~~~\text{and}~~~~~~\mathbb{P}\left[L(Y)-\mathbb{E}[L(Y)]\leq -\lambda\right]\leq\exp\left(-\frac{\lambda^2}{2t^2k}\right).
    \end{align*}
\end{lemma}
To simplify the continued use of \cref{AzumaHoeffding} we consider $k$ to be the number of sites of a sequence alignment $\mathcal{S}$ on $\mathcal{X}$ under the CFN model. Let $Y=(Y_1,\dots,Y_k)$ be independent random variables taking values $Y_{ij}=s_i^j$ for $s^j\in\mathcal{S}$, $i\in\{1,\dots,k\}$ and, for all $i,j\in\{1,\dots,n\}$, define the function $L:\{0,1\}^{n\times k}\to\mathbb{R}$ by
\begin{align*}
    M\mapsto \frac{1}{k}d_H\left(\left(M_{hi}\right)_{h\in\{1,\dots,n\}},\left(M_{hj}\right)_{h\in\{1,\dots,n\}}\right),
\end{align*}
i.e., $|L(M)-L(N)|\leq 1/k$ for all $M,N\in\{0,1\}^{n\times k}$ which differ in only one coordinate. Then, we arrive at following special case of \cref{AzumaHoeffding}:
\begin{corollary}\label{cor::AzumaHoeffding}
    For $x_i,x_j\in\mathcal{X}$, the number of sites $k$ and $\lambda >0$,
    \begin{align*}
        \mathbb{P}\left[h^{ij}-\mathbb{E}[h^{ij}]\geq\lambda\right]\leq\exp\left(-\frac{\lambda^2k}{2}\right)~~~~~~\text{and}~~~~~~\mathbb{P}\left[h^{ij}-\mathbb{E}[h^{ij}]\leq -\lambda\right]\leq\exp\left(-\frac{\lambda^2k}{2}\right).
    \end{align*}
\end{corollary}

Now, by bounding the compound changing probabilities $\mathbb{P}[Y(v_1)\neq Y(v_2)]$ (see \cref{lem::treePath}) and fixing the inheritance probabilities $\gamma(e)$ we can derive a lower bound on the probability that the \lqpc method returns the correct quartet profile.
\begin{proposition}\label{prop::successProb}
Let $\N=(V,E)$ be a semi-directed level-1 network and let $E_{\max}$ be an upper bound on the compound changing probability over all paths in $\N$ between taxa $\{x_i,x_j,x_k,x_l\}$. Suppose $k$ sites evolve under the CFN model on $\N$ and assume $N_{|(x_i,x_j,x_k,x_l)}$ contains a reticulation edge $e$. Then, the probability that \lqpc fails to return the correct quartet profile on taxa $\{x_i,x_j,x_k,x_l\}$ for a fixed inheritance probability $\gamma(e)\in(0,1)$ is at most
\begin{align*}
12\exp\left(- \frac{k}{8}\left(1-\max_{e'\in E}\theta(e')^{\min\{\gamma(e),1-\gamma(e)\}/(2s(e))}\right)^2(1-2E_{\max})^2\right).
\end{align*}
Equivalently, \lqpc returns the correct quartet profile with probability $1-o(1)$ if 
\begin{align}\label{gamma:bounds}
    \min\{\gamma(e),1-\gamma(e)\}\geq 2s(e)\cdot\frac{\log\left(1-\sqrt{M/k}\right)}{\log\max\limits_{e'\in E}\theta(e')}
\end{align}
where $M$ is a sufficiently large positive number increasing in $E_{\max}$ with $k>M$.

\end{proposition}
\begin{proof}
%If $(ijkl)$ is a trivial quartet profile of $\Q_R(\N_{|\{x_i,x_j,x_k,x_l\}})$, then our claim follows from Theorem~8 in~\cite{erdos1999logs1} by trivially setting $\gamma(e)=0$ or $\gamma(e)=1$. 
The quartet profile \niels{$(x_i x_j x_k x_l)\in\Q_R(\N_{|\{x_i,x_j,x_k,x_l\}})$}
%$(ijkl)\in\Q_R(\N_{|\{x_i,x_j,x_k,x_l\}})$ 
is not trivial because $N_{|(x_i,x_j,x_k,x_l)}$ contains a reticulation edge $e$. Let $\gamma$ and $s$ denote $\gamma(e)$ and $s(e)$, respectively. %Recall that lengths of edges $e\in E$ are given by $\ell(e)=-\log\theta(e)/2$. 
Then, we know from \cref{prop::L1FPMsuccess} that L1QPC returns \niels{$(x_i x_j x_k x_l)$} %$(ijkl)$ 
for $\gamma\in(0,1)$ such that for displayed trees $\T_1$ and $\T_2$ of $\N_{|\{x_i,x_j,x_k,x_l\}}$ and $D=\gamma D(\T_1)+(1-\gamma)D(\T_2)$,
\begin{align*}
    \max\left\{|d_r-D_r|\,:\,r\in\{ij,kl,ik,jl\}\right\}
    &<-\frac{\gamma}{4s}\log\max_{e\in E}\theta(e),\\
    \max\left\{|d_r-D_r|\,:\,r\in\{il,jk,ik,jl\}\right\}
    &<\frac{\gamma -1}{4s}\log\max_{e\in E}\theta(e).
\end{align*}
%This means, the probability that L1QPC fails to return quartet profile $(ijkl)$ is at most
%\begin{align*}
    %\mathbb{P}\left[\forall\,\gamma~\exists\,r~:~|d_r-D_r|\geq-\frac{\gamma}{4}\log\max_{e\in E}\theta(e)\right]=1.
%\end{align*}
Then, the probability that \lqpc fails to return the correct quartet profile on taxa $\{x_i,x_j,x_k,x_l\}$ for a fixed inheritance probability $\gamma\in(0,1)$ and $m=\max_{e\in E}\theta(e)$ is upper bounded by
\begin{align*}
&\sum_{r\in\{ij,kl\}}\mathbb{P}\left[|d_r-D_r|\geq-\frac{\gamma}{4s}\log m\right]\\
&~~+\sum_{r\in\{il,jk\}}\mathbb{P}\left[|d_r-D_r|\geq\frac{\gamma -1}{4s}\log m\right]\\
&~~+\sum_{r\in\{ik,jl\}}\mathbb{P}\left[|d_r-D_r|\geq -\frac{\max\{\gamma,1-\gamma\}}{4s}\log m\right].
\end{align*}
Recall that $d_r=-\log(1-2h^r)/2$, $D(\T_i)_r=-\log(1-2\mathbb{E}_i[h^r])/2$, $i\in\{1,2\}$, where $\mathbb{E}_i[h^r]$ denotes the expectation $\mathbb{E}[h^r]$ in $\T_i$. Then, for $r\in\{ij,kl,il,jk,ik,jl\}$,
\begin{align*}
|d_r-D_r|&=\frac{1}{2}\left|\log(1-2h^r)-\gamma\log(1-2\mathbb{E}_1[h^r])-(1-\gamma)\log(1-2\mathbb{E}_2[h^r])\right|\\
&=\frac{1}{2}\left|\log\frac{1-2h^r}{(1-2\mathbb{E}_1[h^r])^{\gamma}(1-2\mathbb{E}_2[h^r])^{1-\gamma}}\right|.%=\frac{1}{2}\log\frac{1-2\min\{h^r,\mathbb{E}[h^r]\}}{1-2\max\{h^r,\mathbb{E}[h^r]\}}.
%=\begin{cases}
%\frac{1}{2}\log\frac{(1-2h^r)^{\alpha_k}}{1-2E^r} &\text{if }h^{ij}\leq E^{ij},\\
%\frac{1}{2}\log\frac{1-2E^r}{(1-2h^r)^{\alpha_k}} &\text{otherwise.}
%\end{cases}
\end{align*}
%Consider one choice of $i,j$ and $D_r$ out of the $12$ possible choices and 
\begin{description}
\item[Case 1:] $d_r\leq D_r$. Without loss of generality $r\in\{ij,kl\}$. Then, we want to calculate
\begin{align}\label{ineq::eventCase1}
    \mathbb{P}\left[\frac{1}{2}\log\frac{1-2h^r}{(1-2\mathbb{E}_1[h^r])^\gamma(1-2\mathbb{E}_2[h^r])^{1-\gamma}}\geq-\frac{\gamma}{4s}\log m\right]
\end{align}
which is equivalent to
\begin{align}\label{ineq::evenCase1.1}
    \mathbb{P}\left[1-2h^r\geq(1-2\mathbb{E}_1[h^r])^\gamma(1-2\mathbb{E}_2[h^r])^{1-\gamma}m^{-\gamma/(2s)}\right].
\end{align}
%\begin{align*}
%    \mathbb{P}\left[(1-2h^{ij})^{1/\gamma}\geq\frac{(1-2\mathbb{E}_1[h^{ij}])(1-2\mathbb{E}_2[h^{ij}])^{(1-\gamma)/\gamma}}{(1-2f)^{1/2}}\right]
%\end{align*}
Let $\mathbb{E}[h^r]=\max\{\mathbb{E}_1[h^r],\mathbb{E}_2[h^r]\}$. Then, probability~\eqref{ineq::evenCase1.1} is upper bounded by
\begin{align*}
    \mathbb{P}\left[1-2h^r\geq(1-2\mathbb{E}[h^r])\,m^{-\gamma/(2s)}\right].
\end{align*}

After shifting the inequality by $(2\mathbb{E}[h^r]-1)$, we arrive at
\begin{align*}
\mathbb{P}\left[h^r-\mathbb{E}[h^r]\leq -\frac{1}{2}\left(m^{-\gamma/(2s)}-1\right)(1-2\mathbb{E}[h^r])\right].
\end{align*}
Since $m^{\gamma/(2s)}\in(0,1)$, we can apply Corollary~\ref{cor::AzumaHoeffding} to upper bound \eqref{ineq::eventCase1} by
\begin{align*}
\exp\left(-\frac{k}{8}\left(m^{-\gamma/(2s)}-1\right)^2(1-2\mathbb{E}[h^r])^2\right).
\end{align*}
%Analogously, for $r\in\{il,jk\}$, we obtain
%\begin{align*}
%    \mathbb{P}\left[|d_r-D_r|\geq\frac{\gamma -1}{4}\log(1-2f)\right]\leq\exp\left(-\frac{k}{8}\left(\frac{1}{(1-2f)^{(1-\gamma)/2}}-1\right)^2(1-2\mathbb{E}[h^r])^2\right).
%\end{align*}
\item[Case 2:] $d_r> D_r$. Without loss of generality $r\in\{ij,kl\}$. Then, we want to calculate
\begin{align}\label{ineq::eventCase2}
    \mathbb{P}\left[\frac{1}{2}\log\frac{(1-2\mathbb{E}_1[h^r])^{\gamma}(1-2\mathbb{E}_2[h^r])^{1-\gamma}}{1-2h^r}\geq-\frac{\gamma}{4s}\log m\right].
\end{align}
Analogous to Case~1, for $\mathbb{E}[h^r]=\max\{\mathbb{E}_1[h^r],\mathbb{E}_2[h^r]\}$, we can upper bound~\eqref{ineq::eventCase2} by
\begin{align}
    \mathbb{P}\left[h^r-\mathbb{E}[h^r]\geq\frac{1}{2}\left(m^{-\gamma/(2s)}-1\right)(1-2h^r)\right].\label{ineq::eventCase2:2}
\end{align}
%$\mathbb{E}[h^{ij}]>h^{ij}$:
%\begin{align*}
%    \mathbb{P}\left[\frac{1}{2}\log\frac{1-2h^{ij}}{1-2E[h^{ij}]}<-\frac{1}{4}\log(1-2f)\right]=\mathbb{P}\left[E[h^{ij}]-h^{ij} < \frac{1}{2}\left(\frac{1}{\sqrt{1-2f}}-1\right)(1-2E[h^{ij}])\right]
%\end{align*}
Let $ \epsilon = \left(1-m^{\gamma/(2s)}\right)/2\in(0,1)$.
\begin{description}
\item[Case 2.1:] $h^r-\mathbb{E}[h^r]\leq\epsilon(1-2\mathbb{E}[h^r])$. Then,
\begin{align*}
(1-2h^r)\geq(1-2\mathbb{E}[h^r])-2\epsilon(1-2\mathbb{E}[h^r])=(1-2\mathbb{E}[h^r])(1-2\epsilon).
\end{align*}
Hence, \eqref{ineq::eventCase2:2} is upper bounded by
\begin{align*}
\mathbb{P}\left[h^r-\mathbb{E}[h^r]\geq \frac{1}{2}\left(m^{-\gamma/(2s)}-1\right)(1-2\epsilon)(1-2\mathbb{E}[h^r])\right].
\end{align*}
Thus, Corollary~\ref{cor::AzumaHoeffding} yields \CorrB{the following upper bound} for~\eqref{ineq::eventCase2}:
\begin{align*}
&\exp\left(- \frac{k}{8}\left(m^{-\gamma/(2s)}-1\right)^2(1-2\epsilon)^2(1-2\mathbb{E}[h^r])^2\right)\\
&=\exp\left(- \frac{k}{8}\left(m^{-\gamma/(2s)}-1\right)^2m^{\gamma/s}(1-2\mathbb{E}[h^r])^2\right).
\end{align*}
\item[Case 2.2:] $h^{ij}-\mathbb{E}[h^r]>\epsilon(1-2\mathbb{E}[h^r])$. \CorrB{Then, we apply Corollary~\ref{cor::AzumaHoeffding}} directly to obtain
\begin{align*}
\mathbb{P}\left[h^r-\mathbb{E}[h^r]\geq\epsilon(1-2\mathbb{E}[h^r])\right]\leq\exp\left(-\frac{k}{2}\epsilon^2(1-2\mathbb{E}[h^r])^2\right).
\end{align*}
\end{description}
Hence, in \CorrB{total,} \eqref{ineq::eventCase2} is upper bounded by
\begin{align*}
2\exp\left(- \frac{k}{8}\left(1-m^{\gamma/(2s)}\right)^2(1-2\mathbb{E}[h^r])^2\right).
\end{align*}
%Analogously, for $r\in\{il,jk\}$, we obtain
%\begin{align*}
%    \mathbb{P}\left[|d_r-D_r|\geq\frac{\gamma -1}{4}\log(1-2f)\right]\leq2\exp\left(- \frac{k}{8}\left(1-(1-2f)^{(1-\gamma)/2}\right)^2(1-2\mathbb{E}[h^r])^2\right).
%\end{align*}
\end{description}
We conclude that the probability of failure for any fixed choice of $\gamma$ is upper bounded by
\begin{align*}
    &\sum_{r\in\{ij,kl\}}2\exp\left(- \frac{k}{8}\left(1-m^{\gamma/(2s)}\right)^2(1-2\mathbb{E}[h^r])^2\right)\\
    &+\sum_{r\in\{il,jk\}}2\exp\left(- \frac{k}{8}\left(1-m^{(1-\gamma)/(2s)}\right)^2(1-2\mathbb{E}[h^r])^2\right)\\
    &+\sum_{r\in\{il,jk\}}2\exp\left(- \frac{k}{8}\left(1-m^{\max\{\gamma,1-\gamma\}/(2s)}\right)^2(1-2\mathbb{E}[h^r])^2\right)\\
    &\leq 12\exp\left(- \frac{k}{8}\left(1-m^{\min\{\gamma,1-\gamma\}/(2s)}\right)^2(1-2E_{\max})^2\right).
\end{align*}
If we upper bound this probability by a small $\epsilon>0$, then we obtain bounds on $\gamma$:
\begin{align*}
    12\exp\left(- \frac{k}{8}\left(1-m^{\min\{\gamma,1-\gamma\}/(2s)}\right)^2(1-2E_{\max})^2\right)&\leq\epsilon\\
    %\Leftrightarrow~~~\frac{k}{8}\left(1-(1-2f)^{\min\{\gamma,1-\gamma\}/2}\right)^2(1-2E_{\max})^2&\geq -\ln\left(\frac{\epsilon}{12}\right)\\
    %\Leftrightarrow~~~\left(1-(1-2f)^{\min\{\gamma,1-\gamma\}/2}\right)^2&\geq \frac{8}{k}\ln\left(\frac{12}{\epsilon}\right)(1-2E_{\max})^{-2}\\
    \Leftrightarrow~~~m^{\min\{\gamma,1-\gamma\}/(2s)}&\leq 1-\sqrt{\frac{8}{k}\ln\left(\frac{12}{\epsilon}\right)(1-2E_{\max})^{-2}}\\
    \text{or}~~~m^{\min\{\gamma,1-\gamma\}/(2s)}&\geq 1+\sqrt{\frac{8}{k}\ln\left(\frac{12}{\epsilon}\right)(1-2E_{\max})^{-2}}.
\end{align*}
Then, for the constant $M:=8\ln(12/\epsilon)(1-2\mathbb{E}_{\max})^{-2}>1$ and $k>M$, we use the fact that $\log m<0$ to arrive at
\begin{align*}
    \min\{\gamma,1-\gamma\}\geq \frac{2s}{\log m}\log\left(1-\sqrt{M/k}\right)~~~\text{or}~~~
    \min\{\gamma,1-\gamma\}\leq \frac{2s}{\log m}\log\left(1+\sqrt{M/k}\right)<0.
\end{align*}
%\begin{align*}
%    \frac{(s+1)\left(1-2s\cdot\frac{\log(1-\sqrt{M/k})}{\log\max_{e\in E}\theta(e)}\right)}{(2s-1)\cdot2s\cdot\frac{\log(1-\sqrt{M/k})}{\log\max_{e\in E}\theta(e)}}\leq\frac{\log\max_{e\in E}\theta(e)}{\log\min_{e\in E}\theta(e)}
%\end{align*}
%\begin{align*}
%    1-2s\cdot\frac{-\log(1-\sqrt{M/k})}{-\log\max_{e\in E}\theta(e)}\leq\frac{(2s-1)\cdot 2s}{s+1}\cdot\frac{-\log(1-\sqrt{M/k})}{-\log\min_{e\in E}\theta(e)}
%\end{align*}
%\begin{align*}
%    \frac{1}{2}\leq s\cdot -\log(1-\sqrt{M/k})\cdot\left( \frac{2s-1}{s+1}\cdot\frac{1}{-\log\min_{e\in E}\theta(e)}+\frac{1}{-\log\max_{e\in E}\theta(e)}\right)
%\end{align*}
%instead
%\begin{align*}
%    -\log(1-\sqrt{M/k})\cdot 2s\cdot \frac{2s-1}{s+1}\to\infty\\
%    -\log(1-\sqrt{M/k})\cdot 2s\to0
%\end{align*}
Thus, our claim follows.
%\begin{align*}
%    &\sum_{\gamma\in(0,0.5]}4\exp\left(- \frac{k}{8}\left(1-(1-2f)^{\gamma/2}\right)^2(1-2E_{\max})^2\right)+8\exp\left(- \frac{k}{8}\left(1-(1-2f)^{(1-\gamma)/2}\right)^2(1-2E_{\max})^2\right)\\
%    +&\sum_{\gamma\in(0.5,1)}8\exp\left(- \frac{k}{8}\left(1-(1-2f)^{\gamma/2}\right)^2(1-2E_{\max})^2\right)+4\exp\left(- \frac{k}{8}\left(1-(1-2f)^{(1-\gamma)/2}\right)^2(1-2E_{\max})^2\right)\\
%    \leq&\sum_{\gamma\in(0,0.5]}24\exp\left(- \frac{k}{8}\left(1-(1-2f)^{\gamma/2}\right)^2(1-2E_{\max})^2\right)\\
%    =&\lim_{M\to\infty}\sum_{\beta=1}^{M}24\exp\left(- \frac{k}{8}\left(1-(1-2f)^{\beta/(4M)}\right)^2(1-2E_{\max})^2\right)\to\infty
%\end{align*}
%Let $a=k(1-2E_{\max})^2/8>0$ and $b=1-2f\in(0,1)$. Then,
%\begin{align*}
%    \sum_{\beta=1}^{M}\exp\left(- a\left(b^{\frac{\beta}{4M}}- 1\right)^2\right)
%\end{align*}
\end{proof}

Finally, we arrive at our main result. \niels{Here, for our choice of function $s$ at the beginning of the section, we denote $s_{\max}$ as the maximum value in the image of $s$. This value increases the closer inheritance probabilities are to binary states. Hence, bounds on the inheritance probabilities yield bounds on $s_{\max}$, too.}

\begin{theorem}\label{thm::CFN}
Suppose $k$ sites evolve under the CFN model on a phylogenetic level-1 network $\N$. Then, the \dc method returns $\NN$ with probability $1-o(1)$, if
\begin{align}\label{kN::LB}
k>\frac{c\cdot\log n}{\left(1-\max\limits_{e\in E}\theta(e)^{\min\{\gamma_{\min},1-\gamma_{\max}\}/s_{\max}}\right)^2\left(\min\limits_{e\in E}\theta(e)\right)^{2(\text{depth}(\N)+lc(\N))}}
\end{align}
where $c$ is a fixed constant, $\gamma_{\min}$ and $\gamma_{\max}$ are the minimum and maximum
attained inheritance probability, respectively. Equivalently, the \dc method returns $\NN$ with probability $1-o(1)$, if inequality~\eqref{gamma:bounds} holds for all inheritance probabilities.
\end{theorem}
\begin{proof}
We show that there exists a set of quartet profiles $\Q_w$ of width $w$ such that $\Q_R(\N)\subseteq\Q_w\subseteq\Q(\N)$ and the \lqpc method returns $\Q_w$ with high probability. To this end, let 
\begin{align*}
    \tau&\in\left(0,\frac{1}{6}\min_{e\in E}\theta(e)^{\text{depth}(\N)+lc(\N)}\right),\\
    S_{\tau}&=\left\{\{x_i,x_j\}\,:\,h^{ij}<\frac{1}{2}-\tau\right\},\\
    Z_{\tau}&=\left\{\{x_i,x_j,x_k,x_l\}\subseteq\X\,:\,h^r<\frac{1}{2}-2\tau~\forall\,r\in\{ij,ik,il,jk,jl,kl\}\right\}.
\end{align*}
Observe that $\{x_i,x_j\}\in S_{2\tau}$ for $x_i,x_j\in Q$, $Q\in Z_{\tau}$, and
\begin{align*}
    &\mathbb{P}\left[\exists\,w\in\{h^{ij}\,:\,1\leq i,j\leq n\}\,:\,\Q_R(\N)\subseteq \Q_w\subseteq\mathcal{Q}(N)\right]\\
    &\geq\mathbb{P}\left[\Q_R(\N)\subseteq Z_{\tau}\text{ and L1QPC returns the correct quartet profile for $Q\in Z_{\tau}$}\right].
\end{align*}
Consider $x_i,x_j\in\mathcal{X}$.
\begin{description}
    \item[Case 1:] $\mathbb{E}[h^{ij}]\geq 1/2-\tau$. Then, from \cref{cor::AzumaHoeffding} it follows
    \begin{align*}
        \mathbb{P}\left[\{x_i,x_j\}\in S_{2\tau}\right]=\mathbb{P}\left[h^{ij}<\frac{1}{2}-2\tau\right]\leq\mathbb{P}\left[h^{ij}-\mathbb{E}[h^{ij}]\leq -\tau\right]\leq \exp\left(-\tau^2k/2\right).
    \end{align*}
    \item[Case 2:] $\mathbb{E}[h^{ij}]<1/2-3\tau$. Then,  from \cref{cor::AzumaHoeffding} it follows
    \begin{align*}
        \mathbb{P}\left[\{x_i,x_j\}\notin S_{2\tau}\right]=\mathbb{P}\left[h^{ij}\geq\frac{1}{2}-2\tau\right]\leq\mathbb{P}\left[h^{ij}-\mathbb{E}[h^{ij}]\geq \tau\right]\leq\exp\left(-\tau^2k/2\right).
    \end{align*}
\end{description}
Hence, we conclude
\begin{align}
    &\mathbb{P}\left[\{x_i,x_j\}\in S_{2\tau}~\forall\,x_i,x_j\in\mathcal{X},\,\mathbb{E}[h^{ij}]<\frac{1}{2}-3\tau\right]+\mathbb{P}\left[\{x_i,x_j\}\notin S_{2\tau}~\forall\,x_i,x_j\in\mathcal{X},\,\mathbb{E}[h^{ij}]\geq\frac{1}{2}-\tau\right]\nonumber\\
    &\geq 1-2\cdot\binom{n}{2}\exp\left(-\tau^2k/2\right).\label{Cbound}
\end{align}
%Let $E_{\min}^{ij}$ and $E_{\max}^{ij}$ denote the expectation $\mathbb{E}[h^{ij}]$ for the topologically shortest and longest $x_i$-$x_j$-path in $\N$, respectively.
As shorthand notation we write $\mathbb{P}[A]$ for the sum of probabilities in~\eqref{Cbound}. Assume that $\{x_i,x_j\}\in S_{2\tau}$ for all $x_i,x_j\in\X$ with $E[h^{ij}]<1/2-3\tau$ in some displayed tree of $N$. This assumption is well founded due to bound~\eqref{Cbound} and implies
\begin{align*}
    h^{ij},E_{\min}^{ij}&<\frac{1}{2}\left(1-\min_{e\in E}\theta(e)^{\text{depth}(\N)+lc(\N)}\right)~&~&\forall\,x_i,x_j\in\X.
\end{align*}
Then, $\Q_R(\N)\subseteq Z_{\tau}$ by \cref{lem::Qw}. Therefore,
\begin{align*}
    &\mathbb{P}\left[\Q_R(\N)\subseteq Z_{\tau}\text{ and L1QPC returns the correct quartet profile for $Q\in Z_{\tau}$}\right]\\
    &\geq\mathbb{P}\left[\Q_R(\N)\subseteq Z_{\tau},\text{ L1QPC returns the correct quartet profile for $Q\in Z_{\tau}$ and event $A$ occurs}\right]\\
    &\geq\mathbb{P}\left[\text{L1QPC returns the correct quartet profile for $Q\in Z_{\tau}$ and event $A$ occurs}\right].
\end{align*}
Now, let $\Q_{\max}$ denote the set of quartets $Q\subseteq\X$ such that $\max\{E[h^{ij}]\,:\,x_i,x_j\in Q\}<1/2-\tau$. 
%Observe that
%\begin{align*}
%    \frac{1}{2}-\tau&>\frac{1}{2}\left(1-\frac{1}{3}\min_{e\in E}\theta(e)^{\text{depth}(\N)+lc(\N)}\right)\\
%    &=\frac{1}{2}\left(1-\min_{e\in E}\theta(e)^{\text{depth}(\N)+lc(\N)-\log(3)/\log\left(\min\{\theta(e)\,:\,e\in E\}\right)}\right)\\
%    &\geq\frac{1}{2}\left(1-\min_{e\in E}\theta(e)^{\text{depth}(\N)+lc(\N)}\right)\geq E_{\max}^{ij}
%\end{align*}
%because $-\log(3)/\log\left(\min\{\theta(e)\,:\,e\in E\}\right)\geq lc(\N)/2$, or equivalently,
%\begin{align*}
%\frac{\log(3)}{lc(\N)}\geq -\frac{1}{2}\log\left(\min\{\theta(e)\,:\,e\in E\}\right)=\min\left\{\ell(e)\,:\,e\in E\right\}
%\end{align*}
%by assumption. 
Without loss of generality \lqpc returns a non-trivial quartet profile $(x_ix_jx_kx_l)$ for $Q$. For a reticulation edge $e$ in $N_{|\{x_i,x_j,x_k,x_l\}}$, \niels{$(x_i x_j x_k x_l)$} %$(ijkl)$ 
is the correct quartet profile with probability at least
\begin{align}\label{ineq::nonTrivialsuccess}
    1-12\exp\left(- \frac{k}{8}\left(1-\max_{e'\in E}\theta(e')^{\min\{\gamma(e),1-\gamma(e)\}/(2s(e))}\right)^2\tau^2\right)
\end{align}
by \cref{prop::successProb}. Otherwise, $(x_ix_jx_kx_l)$ is trivial and Theorem~8 in~\cite{erdos1999logs1} applies which equates to setting $\gamma(e)=s(e)=1$ in formula~\eqref{ineq::nonTrivialsuccess}. Therefore, we can apply bound~\eqref{Cbound} to infer that $Z_r\subseteq\Q_{\max}$ with high probability. Hence,
\begin{align*}
&\mathbb{P}\left[\text{L1QPC returns the correct quartet profile for $Q\in Z_{\tau}$ and event $A$ occurs}\right]\\
    &\geq\mathbb{P}\left[\text{L1QPC returns the correct quartet profile for $Q\in Q_{\max}$ and event $A$ occurs}\right]\\
    &\geq 1-12\exp\left(- \frac{k}{8}\left(1-\max_{e\in E}\theta(e)^{\min\{\gamma_{\min},1-\gamma_{\max}\}/s_{\max}}\right)^2\tau^2\right)-(n^2-n)\exp\left(-\frac{k}{2}\tau^2\right)
\end{align*}
where the second inequality follows from the Bonferonni inequality. Thus, there exists a constant $c$ such that
\begin{align*}
    &\mathbb{P}\left[\exists\,w\in\{h^{ij}\,:\,1\leq i,j\leq n\}\,:\,Q_R(N)\subseteq Q_w\subseteq\mathcal{Q}(N)\right]\\
    &\geq 1-n^c\exp\left(-\frac{k}{8}\left(1-\max_{e\in E}\theta(e)^{\min\{\gamma_{\min},1-\gamma_{\max}\}/s_{\max}}\right)^2\tau^2\right).
\end{align*}
\end{proof}

\begin{corollary}\label{cor:CFN}
    Suppose $k$ sites evolve under the CFN model on a phylogenetic level-1 network $\N$ with fixed inheritance probabilities. Then, for depth$(N)=\mathcal{O}(\log\log n)$, a constant upper bound on the number of edges in each cycle in $N$ and constant bounds on the mutation probabilities, the \dc method returns $\NN$ with probability~$1-o(1)$, if $k=\,\text{polylog}(n)$.
\end{corollary}
\begin{proof}
\CorrB{
    \cref{thm::CFN} tells us that, for a fixed constant $c$,
    \begin{align}\label{corBound}
       \frac{c\cdot\log n}{\left(1-\max\limits_{e\in E}\theta(e)^{\min\{\gamma_{\min},1-\gamma_{\max}\}/s_{\max}}\right)^2\left(\min\limits_{e\in E}\theta(e)\right)^{2(\text{depth}(\N)+lc(\N))}}
    \end{align}
    is a strict lower bound on the given number of sites $k$ for the \dc method to return $\NN$ with high probability. Since we assume that inheritance probabilities are fixed, $\gamma_{\min}$ and $\gamma_{\max}$ satisfy inequality~\eqref{gamma:bounds}. From \cref{prop::successProb} we know that the inheritance probabilities adhering to inequality~\eqref{gamma:bounds} are fully characterized by parameters $k$, $E_{\max}$ and $\theta(e)=1-2p(e)$, $e\in E(N)$. The latter two are bounded by constants because we assume that the mutation probabilities $p(e)$ admit constant bounds. Recall from the beginning of this section that $s_{\max}$ is a function of probabilities $p(e)$ and $\gamma(e)$, $e\in E(N)$, implying that $s_{\max}$ is constantly bounded, too. Finally, we assume depth$(N)=\mathcal{O}(\log\log n)$ and that lc$(N)$ is constant. This means, for suitable constants $a$ and $b$, bound~\eqref{corBound} simplifies to
    \begin{align*}
        \frac{c\cdot\log n}{a^2\cdot b^{\mathcal{O}(\log\log n)}}=\frac{c\cdot\log n}{a^2\cdot (\log n)^{\mathcal{O}(1)}}=\mathcal{O}\left((\log n)^{c'}\right)
    \end{align*}
    for some constant $c'$.
}
\end{proof}

\cref{tab::2} in the introduction broadens \cref{cor:CFN} to further conditions to derive lower bounds on the number of sites of a sequence alignment from \cref{thm::CFN}. \CorrB{For example, if we assume that our measures depth$(N)$ and lc$(N)$ of the network topology are independent of the number of taxa, i.e., are bounded by constants, then \cref{cor:CFN} can be strengthened to only require a logarithmic sequence length. In other words, the number of sites sufficient to reconstruct the correct binary level-1 phylogenetic network is logarithmic in the number of taxa of the network. The same argument can be extended in the opposite direction: if depth$(N)$ and lc$(N)$ grow logarithmically in the number of taxa, then \cref{cor:CFN} can be relaxed to conclude that the sufficient sequence length is a polynomial in the number of taxa. \cref{tab::2} details further bounds in case we vary the bounds on mutation probabilities, too. Observe that the bounds on inheritance probabilities are solely defined by the bounds on mutation probabilities. Hence, to derive conclusions on the sufficient sequence length we do not consider any further assumptions on the inheritance probabilities.}

\section{Discussion}\label{sec:discussion}

In this manuscript we have provided a first analysis \CorrB{of} the amount of genomic data required to accurately reconstruct phylogenetic networks, \leo{without a homoplasy-free restriction.}
%in the presence of homoplastic sites.
While our results are based on a fairly simple model of evolution and %our bounds
\leo{depend on bounds on}
 several parameters, they offer an optimistic outlook, showing that the required sequence length scales logarithmically, polynomially, or polylogarithmically with the number of taxa. In particular, these results suggest that accurate inference of evolutionary histories in the presence of reticulate events is possible in certain cases, even with a restricted amount of genomic data available. Moreover, we hope to have  established a necessary bridge between convergence questions in the extensively studied context of tree inference and their mostly unstudied counterparts in network inference. This connection opens the door to rigorous analyses of broader classes of networks and more complex evolutionary models. It may also lead to new insights into the behavior and reliability of existing software tools for network reconstruction.

Although the analysis of our \dc method assumes a CFN model of evolution, the method itself can be applied much more generally, relying only on a distance matrix between taxa as input. Hence, as a first step, one might extend our analysis to related models of evolution, such as the Jukes-Cantor, Kimura-2-Parameter and Kimura-3-Parameter models, or possibly models allowing for ILS (incomplete lineage sorting). Another direction could be to instead consider different types of distances that can be obtained from genomic data (see e.g. \cite{allman2019nanuq,allman2024nanuq+, xu2023identifiability}).

Extending our results to networks more general than level-1 will require additional work on the algorithm itself. Although most of our proposed inference rules are valid for networks of any level, we have not established closure properties---a key component of our method---when going beyond level-1. In particular, we do not expect our current set of inference rules to suffice, since networks beyond level-1 %may not be 
are not all
outer-labeled planar %This implies that 
% Leo: does not-outer-labelled planar imply size-3 quartet profiles?
% Niels: Good point, I don't think this has been shown yet, although I expect this to be true.
and some quartet profiles could contain all three possible quartet trees on four leaves, a case that is not covered by our current set of rules. A %more
natural next step may %instead
be to first study the tree-of-blobs of a network---the tree obtained by contracting each %nontrivial
blob to a single vertex---which could serve as a bridge to more general network classes. It must be noted that one cannot simply rely on existing theory for trees however, since a trivial quartet profile of a network is not necessarily induced by its tree-of-blobs (see also the distinction between \emph{T-quartets} and \emph{B-quartets} from \cite{allman2023tree}).

We end by noting that
%a closer analysis\todo{shown where? NH: I explained this somewhere in the text at some point. However, I don't think we can make any claims regarding the closure etc., simply that our rules can be extended to the case where you e.g. do know the reticulation in 4-cycles. Might need to rephrase.} of our proofs reveals that
our inference rules naturally extend to semi-directed quarnets (4-leaf subnetworks).
Closure properties for this case may be an interesting direction for future research (as studied for undirected quarnet inference rules in \cite{huber2018quarnet}).
%thus forming a complementary counterpart to 
%\leo{related to} the undirected quarnet %framework
%inference rulues studied in \cite{huber2018quarnet},
%\leo{and closure properties of }

\paragraph{Funding.} This work was partially supported by grant~OCENW.M.21.306 of the Dutch Research Council~(NWO).

\paragraph{Competing interests.} On behalf of all authors, the corresponding author states that there is no conflict of interest.

\paragraph{Data availability.} No datasets were generated or analysed for this study. All information needed to evaluate the findings is contained within the manuscript.

%\bibliographystyle{plain}
%\bibliography{references}
%\printbibliography[title = References]

\appendix

\section{Displaying trees and inducing subnetworks are commutative}\label{sec:commutative}

\CorrB{In the following proposition we prove that taking displayed trees and taking subnetworks are commutative operations. That is, the set of displayed trees of a subnetwork is the same as the set of induced subtrees of the displayed trees of the full network.}

\begin{proposition}\label{prop:displayed_subnet_commutative}
    \CorrB{Let $\N$ be a semi-directed network on $\X$ and let $\T$ be an unrooted phylogenetic tree on $\Y \subseteq \X$ with $|\Y| \geq 2$. Then, $\T$ is a displayed tree of $\N|_\Y$ if and only if $\T$ is a subtree induced by $\Y$ of some displayed tree on $\X$ of $\N$.}
\end{proposition}
\begin{proof}
\CorrB{First note that each reticulation vertex~$r$ of $\N|_\Y$ corresponds to a unique reticulation vertex~$r'$ in $\N$. We denote the set of these reticulation vertices in $\N$ by $R$. Let $P$ be the set of reticulation vertices of $\N$ not in $R$. Since we allow parallel edges, a reticulation vertex~$r$ of $\N$ is in $R$ if and only if it is on an up-down path in $\N$ between two leaves in~$\Y$.}

\CorrB{Suppose that $\T$ is a displayed tree of $\N|_\Y$. Without loss of generality, suppose we label the edges such that $\T$ can be obtained from $\N|_\Y$ by deleting the reticulation edge $e_i^1$ for each pair of reticulation edges $(e_i^1, e_i^2)$ (with reticulation vertex $r_i \in R$) in $\N|_\Y$. Let $\T'$ be the displayed tree of $\N$ obtained by removing the edge $e_i^1$ for each $r'_i \in R$, and arbitrarily removing one reticulation edge for each reticulation vertex in $P$. Then, $\T$ is
an induced subtree
\leo{$\T'|_\Y$}
of $\T'$,
proving the forward implication.}

%For the backward implication, note that if $\T$ is an induced subtree of some displayed tree $\T'$ of $\N$, then each edge of $\T$ corresponds to an up-down path (between leaves of $\Y$) in $\T'$ and hence an up-down path (between leaves of $\Y$) in $\N$. So, $\T$ is a displayed tree of $\N|_Y$.

\CorrB{For the backward implication, suppose that $\T$ is an induced subtree of some displayed tree $\T'$ of $\N$. Without loss of generality, assume that we label the edges such that $\T'$ can be obtained from $\N$ by deleting the reticulation edge $e_i^1$ for each pair of reticulation edges $(e_i^1, e_i^2)$ (with reticulation vertex $r_i \in R \cup P$) in $\N$. Similar to before, $\T$ can then be obtained from $\N|_\Y$ by removing the reticulation edge $e_i^1$ for each $r_i \in R$. This proves the other implication.}
\end{proof}

\section{Proof of Lemma 5.1}\label{sec:appendix}

Recall, for a semi-directed level-1 network $\N$ under the CFN model we have a well-defined function $s$ which maps the reticulation edges $e$ of $\N$ to the positive real numbers by equations
\begin{align*}
    \frac{\gamma(e)}{1-\gamma(e)}\cdot\frac{\log\min_{e'\in E}\theta(e')}{\log\max_{e'\in E}\theta(e')}+\epsilon&=
    \begin{cases}
        \frac{2s(e)^2-7s(e)}{2s(e)^2+s(e)-1} &\text{if } \frac{\gamma(e)}{1-\gamma(e)}\cdot\frac{\log\min_{e'\in E}\theta(e')}{\log\max_{e'\in E}\theta(e')}<1\\
        \frac{4s(e)^3+4s(e)^2-s(e)-1}{12s(e)^2+4s(e)+1} &\text{otherwise}
    \end{cases}
\end{align*}
for a sufficiently small $\epsilon>0$ such that $s(e)>1/2$.

\lemLFPMsuccess*
\begin{proof}
Let $\gamma$ and $s$ denote the inheritance probability and function $s(e)$ for the reticulation edge $e$ from $N$ present in $\T_1$.
%$D$ circulant:
%\begin{align*}
%    D_{ik}+D_{jl}&>\max\{D_{ij}+D_{kl},D_{il}+D_{jk}\}\\
%    \Leftrightarrow~~1&>\max\left\{\frac{D_{ij}+D_{kl}}{D_{ik}+D_{jl}},\frac{D_{il}+D_{jk}}{D_{ik}+D_{jl}}\right\}
%\end{align*}
Let $\ell_1,\ell_2,\ell_3,\ell_4$ be the edge lengths of the internal edges $e_1,e_2,e_3,e_4$ in $\N$ on paths from $x_i$ to $x_j$, from $x_j$ to $x_k$, from $x_k$ to $x_l$ and from $x_l$ to $x_i$, respectively. 
%\begin{description}
%\item[Case 1:] 
Without loss of generality $e_3\in E(\T_1)\setminus E(T_2)$ and $e_4\in E(\T_2)\setminus E(T_1)$. Assume $\gamma\in(0,1)$ such that $D=\gamma D(\T_1)+(1-\gamma)D(\T_2)$ and conditions~\eqref{L1QPCsuccess::con1} and~\eqref{L1QPCsuccess::con2}.
Then, \niels{$(x_i x_j x_k x_l) \in\mathcal{Q}_R(\N)$} %$(ijkl)\in\mathcal{Q}_R(\N)$ 
only if
\begin{align}
D_{ik}+D_{jl} - \left(D_{ij} + D_{kl}\right) &=\gamma(\ell_1+2\ell_2+\ell_3)+(1-\gamma)(2\ell_1+\ell_2+\ell_4)\nonumber\\
&~~-\left[\gamma(\ell_1+\ell_3)+(1-\gamma)(2\ell_1+\ell_2+\ell_4)\right]\nonumber\\
&=(2-\gamma)\ell_1+(\gamma +1)\ell_2+\gamma \ell_3+(1-\gamma)\ell_4\nonumber\\
&~~-\left[(2-\gamma)\ell_1+(1-\gamma)\ell_2+\gamma \ell_3+(1-\gamma)\ell_4\right]=2\gamma \ell_2,\label{con::diff1}\\
D_{ik}+D_{jl} - \left(D_{il} + D_{jk}\right) &=\gamma(\ell_1+2\ell_2+\ell_3)+(1-\gamma)(2\ell_1+\ell_2+\ell_4)\nonumber\\
&~~-\left[\gamma(\ell_1+2\ell_2+\ell_3)+(1-\gamma)(\ell_2+\ell_4)\right]\nonumber\\
&=(2-\gamma)\ell_1+(\gamma+1)\ell_2+\gamma \ell_3+(1-\gamma)\ell_4\nonumber\\
&~~-\left[\gamma \ell_1+(\gamma +1)\ell_2+\gamma \ell_3+(1-\gamma)\ell_4\right]=2(1-\gamma)\ell_1,\label{con::diff2}\\
D_{il}+D_{jk}-(D_{ij}+D_{kl})&=\gamma(\ell_1+2\ell_2+\ell_3)+(1-\gamma)(\ell_2+\ell_4)\nonumber\\
&~~-\left[\gamma(\ell_1+\ell_3)+(1-\gamma)(2\ell_1+\ell_2+\ell_4)\right]\nonumber\\
&=2\gamma\ell_2-2(1-\gamma)\ell_1.\label{con::diff3}
\end{align}
Without loss of generality $d_{ij}+d_{kl}\leq d_{il}+d_{jk}$. Then, our claim follows if we show that
\begin{align*}
    d_{il}+d_{jk}-d_{ij}-d_{kl}&<d_{ik}+d_{jl}-d_{il}-d_{jk}
\end{align*}
\CorrB{Equivalently, we show that $2(d_{il}+d_{jk})<d_{ik}+d_{jl}+d_{ij}+d_{kl}.$ To this end, for}
%\begin{align*}
%    d_{il}+d_{jk}&>D_{il}+D_{jk}-\frac{1-\gamma}{s}f\\
%    &=D_{ij}+D_{kl}+2\gamma\ell_2-2(1-\gamma)\ell_1-\frac{1-\gamma}{s}f\\
%    &>d_{ij}+d_{kl}+\left(2-\frac{1}{s}\right)\left(\gamma\ell_2-(1-\gamma)\ell_1\right)
%\end{align*}
$f=-\frac{1}{2}\log\max_{e\in E}\theta(e)\leq-\frac{1}{2}\log(1-2p(e))$, $e\in E$, we get
\begin{align*}
    d_{ij}+d_{kl}&\stackrel{\eqref{L1QPCsuccess::con1}}{>}D_{ij}+D_{kl}-\frac{\gamma}{s}f\\
    &\stackrel{\eqref{con::diff3}}{=}D_{il}+D_{jk}-2\gamma\ell_2+2(1-\gamma)\ell_1-\frac{\gamma}{s}f\\
    &\stackrel{\eqref{L1QPCsuccess::con2}}{>}d_{il}+d_{jk}+\left(2-\frac{1}{s}\right)(1-\gamma)\ell_1-\left(2+\frac{1}{s}\right)\gamma\ell_2.
\end{align*}
Since $d_{ij}+d_{kl}\leq d_{il}+d_{jk}$, it follows \CorrB{that}
\begin{align}
\left(2-\frac{1}{s}\right)(1-\gamma)\ell_1<\left(2+\frac{1}{s}\right)\gamma\ell_2.\label{ineq::wlog}
\end{align}
Moreover, we have
\begin{align*}
d_{ik}+d_{jl}+d_{ij}+d_{kl}&\stackrel{\ref{L1QPCsuccess::con1}}{>}D_{ik}+D_{jl}-\frac{\gamma}{s} f+D_{ij}+D_{kl}-\frac{\gamma}{s} f\\
    &\stackrel{\eqref{con::diff2},\eqref{con::diff3}}{=}D_{il}+D_{jk}+2(1-\gamma)\ell_1+D_{il}+D_{jk}-2\gamma\ell_2+2(1-\gamma)\ell_1-\frac{2\gamma}{s} f\\
    &=2\left(D_{il}+D_{jk}+2(1-\gamma)\ell_1-\gamma\ell_2-\frac{\gamma}{s} f\right)\\
    &\geq2\left(D_{il}+D_{jk}+\frac{1-\gamma}{s}f+\left(2-\frac{1}{s}\right)(1-\gamma)\ell_1-\left(1+\frac{1}{s}\right)\gamma\ell_2\right)\\
    &\stackrel{\ref{L1QPCsuccess::con2}}{>} 2\left(d_{il}+d_{jk}+\left(2-\frac{1}{s}\right)(1-\gamma)\ell_1-\left(1+\frac{1}{s}\right)\gamma\ell_2\right).
\end{align*}
Hence, it is left to show that
\begin{align*}
    \left(2-\frac{1}{s}\right)(1-\gamma)\ell_1\geq\left(1+\frac{1}{s}\right)\gamma\ell_2
\end{align*}
\CorrB{Equivalently, we show that
\begin{align}
\frac{\gamma}{1-\gamma}\cdot\frac{\ell_2}{\ell_1}\leq\frac{2s-1}{s+1}.\label{ineq::sL}
\end{align}}
Since inequality~\eqref{ineq::wlog} can be rewritten as
\begin{align*}
    \frac{\gamma}{1-\gamma}\cdot\frac{\ell_2}{\ell_1}>\frac{2s-1}{2s+1}~~~\text{or}~~~\frac{1-\gamma}{\gamma}\cdot\frac{\ell_1}{\ell_2}<\frac{2s+1}{2s-1},
\end{align*}
we know that inequality~\eqref{ineq::sL} can hold only if
\begin{align*}
    \frac{2s-1}{s+1}-\frac{2s-1}{2s+1}> 0.
\end{align*}
In other words, we require $s>1/2$. Moreover, we observe that inequality~\eqref{ineq::sL} holds if
\begin{align}
    \underbrace{\frac{\gamma}{1-\gamma}\cdot\frac{\ell_2}{\ell_1}}_{=:L(\gamma,\ell_1,\ell_2)}-\frac{1-\gamma}{\gamma}\cdot\frac{\ell_1}{\ell_2}&\leq\underbrace{\frac{\gamma}{1-\gamma}\cdot\frac{\ell_2}{\ell_1}-\frac{2s-1}{2s+1}}_{>0}+\underbrace{\frac{2s-1}{s+1}-\frac{2s+1}{2s-1}}_{=:R(s)}\label{con::proof1}\\
    \Leftrightarrow~~~\frac{1}{L(\gamma,\ell_1,\ell_2)}&\geq \frac{2s-1}{2s+1}-R(s).\label{con::proof2}
\end{align}
Then, for $L(\gamma,\ell_1,\ell_2)<1$, we have
\begin{align*}
    \lim_{s\to\infty}R(s)=\lim_{s\to\infty}\frac{2s^2-7s}{2s^2+s-1}=1>L(\gamma,\ell_1,\ell_2),
\end{align*}
implying inequality~\eqref{con::proof1}. Otherwise,
%$L(\gamma)\geq 1$, it follows
%\begin{align*}
%    L(\gamma)\geq 1~~~&\Leftrightarrow~~~\frac{\gamma}{1-\gamma}\cdot\frac{\ell_2}{\ell_1}\geq 1+\frac{1-\gamma}{\gamma}\cdot\frac{\ell_1}{\ell_2}\\
%    &\Leftrightarrow~~~1\geq \frac{1-\gamma}{\gamma}\cdot\frac{\ell_1}{\ell_2}+\left(\frac{1-\gamma}{\gamma}\cdot\frac{\ell_1}{\ell_2}\right)^2\\
%    &\Leftrightarrow~~~\frac{1-\gamma}{\gamma}\cdot\frac{\ell_1}{\ell_2}\in\left(0,\frac{2}{1+\sqrt{5}}\right]
%\end{align*}
inequality~\eqref{con::proof2} follows from
\begin{align*}
    \lim_{s\to\infty}\frac{2s-1}{2s+1}-R(s)=\lim_{s\to\infty}\frac{12s^2+4s+1}{4s^3+4s^2-s-1}=0<\frac{1}{L(\gamma,\ell_1,\ell_2)}.
\end{align*}
Thus, since $$L(\gamma,\ell_1,\ell_2)\leq \frac{\gamma}{1-\gamma}\cdot\frac{\log\min_{e\in E}\theta(e)}{\log\max_{e\in E}\theta(e)},$$ our claim follows from the definition of $s$.
\end{proof}

\end{document}